\documentclass[12pt]{article}

\usepackage{amsfonts}
\usepackage{eurosym}
\usepackage{graphicx}
\usepackage{booktabs}
\usepackage{float}
\usepackage{amssymb}
\usepackage{comment}
\usepackage{amsmath}
\usepackage{mathrsfs}
\usepackage{enumitem}
\usepackage{booktabs}
\usepackage{multirow,array}
\usepackage{color}
\usepackage[authoryear]{natbib}
\usepackage{hyperref}
\usepackage{rotating}
\usepackage{accents}
\usepackage{bm}
\usepackage{bbm}
\usepackage{mathtools}
\usepackage{tikz}
\usetikzlibrary{patterns}
\usetikzlibrary{decorations.pathreplacing}
\usepackage{subcaption}
\usepackage{breqn}
\usepackage{amsthm}
\usepackage{xcolor}
\usepackage{setspace}

\newcommand\supp{{\mathrm{supp}}}

\hypersetup{colorlinks=true, linkcolor=blue, citecolor=gray, urlcolor=gray}

\newtheorem{lemma}{{\bf \sc Lemma}}
\newtheorem{example}{{\bf \sc Example}}
\newtheorem*{continuancex}{Continuation of Example \continuanceref}
\newenvironment{continuance}[1]
  {\newcommand\continuanceref{\ref{#1}}\continuancex}
  {\endcontinuancex}

\newtheorem{proposition}{{\bf \sc Proposition}}

\newtheorem{observation}{{\bf \sc Observation}}

\newtheorem*{properties}{{\bf \sc Properties}}

\DeclareMathOperator*{\argmin}{arg\,min}
\DeclareMathOperator*{\argmax}{arg\,max}

\graphicspath{ {./images/} } 

\begin{document}

\title{Interactions across multiple games: \\
cooperation, corruption, and organizational design}
\author{Jonathan Bendor, Lukas Bolte, Nicole Immorlica, Matthew O. Jackson\thanks{Jonathan Bendor, 
Stanford University;
Lukas Bolte, 
Tilburg University; 
Nicole Immorlica, 
Yale University and 
Microsoft Research; 
Matthew O. Jackson, 
Stanford University 
and the 
Santa Fe Institute. 
Jonathan Bendor passed away in November 2025.  His intuition, deep insights, and wide-ranging perspective were vital in the development of this project, and we were very fortunate to have had him as a coauthor and friend.
We gratefully acknowledge support under NSF grant  SES-2018554. We thank seminar participants from the World Congress of the Game Theory Society, the European Winter Meetings of the Econometric Society, the International Workshop on Economic Theory, VSET, as well as Antonio Cabrales, Wei Li, Roger Myerson, and Alex Wolitzky for helpful comments and discussion.}}

\date{February 2026 
}

\maketitle

\begin{abstract}
Teamwork is vital in many settings, and it is socially beneficial for teams to cooperate in some situations (``good games'') and not in others (``bad games;'' e.g., those that allow for corruption). A team's cooperation in any given game depends on expectations of cooperation in future iterations of both good and bad games. We identify when sustaining cooperation on good games necessitates cooperation on bad games. We then characterize how a designer should optimally assign workers to teams and teams to tasks that involve varying arrival rates of good and bad games. Our results show how organizational design can be used to promote cooperation while minimizing corruption. 


\end{abstract}

\setcounter{page}{0}\thispagestyle{empty}
\newpage

\section{Introduction}

Teamwork---informal cooperation among members---is essential in almost all forms of production, both private and public. 
However, team members can cooperate in ways that are harmful to a firm, organization, or society.  Importantly, ``good'' and ``bad'' forms of cooperation are intertwined. For example, building a cooperative culture among police not only induces productive cooperation, such as aiding other officers in dangerous circumstances, but can also lead them to cooperate in corrupt or destructive ways, such as the ``blue wall of silence'' (police officers under-reporting misconduct).\footnote{See \cite{chevigny1995}. Similarly, some scholars (e.g., \cite{gibson2003wall}) document a ``white wall of silence'': doctors not reporting colleagues' malpractice.} If a police officer does not collude with a partner who has taken a bribe or used excessive force, that partner may not back the officer up in a future dangerous situation.\footnote{Famously, in 1970, Frank Serpico, then a police detective in the NYPD, exposed police corruption by contributing to a New York Times article. A year later, after being shot and wounded during a drug raid, Serpico alleged that his police partners led him into a dangerous situation to be murdered \citep{phalen2012}. Ostracism of whistleblowers is anticipated and widespread: ``Officers who report misconduct are ostracized and harassed; become targets of complaints and even physical threats; and are made to fear that they will be left alone on the streets in a time of crisis. This draconian enforcement of the code of silence fuels corruption because it makes corrupt cops feel protected and invulnerable'' (\cite{mollen1994}).} 
Similar interactions between good and bad forms of cooperation arise in a wide variety of organizations, including in government, military, hospitals, consulting, construction, manufacturing, and retail.
  
We study when and how the structure and design of organizations can encourage socially beneficial forms of cooperation while inhibiting harmful forms of cooperation. We model teams that face a random arrival of games over time. When a game arrives, each team member chooses whether to cooperate with others. In any given game, a team member prefers overall team cooperation to a lack of cooperation but individually prefers not to cooperate. Thus, the games have a prisoners' dilemma-like structure. 

This is not a simple repeated game: the stage games take on different forms. In some of the games, cooperation is beneficial to society (``good games''); in other games, cooperation is harmful (``bad games''). Socially optimal equilibria are those in which team members cooperate in good games only. Yet as \cite{tirole88} conjectured, ``\dots it may be worth tolerating some minor acts of detrimental co-operation to allow trust between employees to develop.'' We show that the team-design problem is nontrivial because of interdependencies between the different stage games: equilibrium structure depends on the arrival frequencies of both good and bad games as well as their payoff structures. 

Our analysis takes a design perspective.  By adjusting team composition and assignments to various tasks, an organizational designer can control the incentives and behavior of the agents involved without observing actions in specific games or even whether a game occurs.  This opens a new set of mechanism design problems, in which the dynamic set faced is controlled by the designer.   This design perspective involves several dimensions. 

One design dimension that we explore is team composition and durability. Some organizations, such as the military or police, keep teams together over time to foster cooperation. Others---e.g., diplomacy, retail, logistics, and services---frequently reshuffle members across teams or simply hire short-term contract workers for certain tasks, eliminating the repeated interaction necessary for cooperation. In which settings should teams be kept together versus reshuffled?  Again, this issue is nontrivial  as \cite{tirole88} anticipated: ``Fighting the detrimental cooperation by keeping the employees' relationship short may also kill the beneficial co-operation.''  We formalize this insight and identify conditions that make detrimental cooperation necessary to sustain socially beneficial cooperation.  

A second dimension that we analyze is team specialization. Organizations typically have a variety of tasks that must be completed, and each task faces inherently different rates of good and bad games. In some organizations, different teams work on different tasks; in others, all teams work on the same tasks. When teams are specialized, some more frequently find themselves in situations where cooperation is socially beneficial (e.g., dangerous situations where police officers need back-up), while other teams more often face opportunities to collude (e.g., to accept bribes). We study when it is optimal to specialize teams, essentially creating two parts of the organization, as opposed to when optimality requires leveraging substantial interdependencies between good and bad cooperation and precluding specialization. 

We also consider dynamic and reactive versions of the assignment of teams to tasks---taking full advantage of the dynamic aspect of dynamic design. Teams may be completing a certain task for some duration of time before being moved to another; they may also be moved in response to, say, an opportunity to cooperate in a situation where it is socially beneficial, e.g., after police officers find themselves in a dangerous spot. Dynamically and reactively reassigning teams to tasks further allows the organization to affect the random stream of opportunities to cooperate. We fully characterize the optimal structure of such assignments for a special case---when it is sufficiently easy to sustain cooperation in the bad game. 

As a basis for these analyses, we first characterize the player-optimal equilibrium for teams facing a fixed distribution of tasks as a function of the payoffs to cooperation and defection in both good and bad games, the games' arrival probabilities, and a team's durability (how frequently team memberships are reshuffled). 

Equipped with the characterization of equilibrium play, we then ask when it is possible to reshuffle teams facing a fixed distribution of tasks at a rate so that team members cooperate in good games only. Our first result identifies three conditions: 
\begin{itemize}
\item[(i)] it is more tempting to deviate in bad games than in good ones,
\item[(ii)] good games are sufficiently likely to arrive (in absolute terms) to sustain cooperation on them alone, and
\item[(iii)] good games are sufficiently likely to arrive (in relative terms) to permit an optimal level of team durability that enables cooperation in good games but not in bad ones.
\end{itemize}

If all three conditions hold, to achieve teams cooperating only in good games, partial reshuffling of teams may be required---occasionally reassigning team members may be socially optimal. If team members expect enough---but not too much---future interactions with their teammates, then cooperation can be supported in one game but not the other.  

If any of these three conditions fails, the organization is faced with a difficult choice: either keep teams together and end up with cooperation in all games (including bad ones) or reshuffle teams enough to preclude all kinds of cooperation. Hence, in some settings, corruption may be a necessary price for inducing productive cooperation. For example, in policing or certain types of military activities, the benefits of cooperation in good situations may be high---e.g., supporting one's team saves lives---relative to the damage caused by cooperation in bad ones---e.g., accepting bribes. However, the temptation to deviate from cooperation may also be high as agents may need to risk their lives to save teammates. Thus, cooperation in good situations can only be sustained with durable teams and may \emph{require} cooperation in bad situations. In contrast, consider warehouse workers fulfilling orders. The organization may rarely need employees to collaborate; the more frequent and costly concern is that workers can collude to steal merchandise. In this situation, ``teams'' are reshuffled frequently, creating a more anonymous workforce with no cooperation---effectively, no teams.

We then consider situations in which different teams can be assigned to different tasks, affecting the likelihood with which they face good and bad games. The designer can then keep certain teams (those facing mostly good games) together to encourage cooperation while reshuffling others (those facing mostly bad games) to deter collusion. For instance, the chief of a large police force might have some units respond to dangerous armed situations while others deal with tasks such as routine traffic stops that allow for bribes. In such circumstances, only the former teams would stay together over time. For example, we show that if there are two tasks, one with only bad games and one with only good games, and facing only good games is enough to sustain cooperation, then fully specializing teams is optimal. In this case,  interdependencies across games are removed. However, when this is not the case and cooperation from only facing good games is infeasible, then the designer may benefit from allowing for some interdependencies, i.e., having teams face both games to allow for cooperation, including socially undesirable. 

We also allow the assignment of teams to tasks to be dynamic and reactive---i.e., the current task assignment can be a function of past assignments. For example, a team may be rotating between staying on a task with only good games (for some time) and staying on a task with only bad games (for some time). Or a team may be randomly moved between tasks. We also consider the case where the designer observes whether a game occurred (but not the equilibrium play) and can condition future task assignments on that. The space of such ``reactive assignments'' is large, and so we focus on a special case where cooperation on bad games is easy to sustain. In an optimal assignment, when the designer observes whether a good game occurred, teams are assigned to the task with more good games until a good game occurs, at which point they are reassigned to tasks with mostly bad games for a number of periods before returning. We derive a similar result for the case where the designer does not observe good games: the difference is that teams move on from the task with good games whether or not a good game occurred. 

In the online appendix, we provide comparative static results with regard to the payoff structure. If an organization can directly affect payoffs, which is not our main focus, then these comparative statics shed light on other strategies for enabling beneficial cooperation. A designer can encourage good cooperation while deterring bad cooperation by altering the benefits of cooperation in good games (for instance, giving bonuses for overall profits) or increasing the temptations of defecting in bad games (for instance, offering whistleblower rewards). Interestingly, these behave quite differently than in standard repeated games, where only the ratio of benefits to temptations matters. With multiple games, benefits from cooperation in a given game have spillovers that enable cooperation in \emph{all} games, while temptations to deviate affect \emph{only one} game's incentives at a time. Thus, the latter is more selective than the former. The effective costs of altering the benefits from cooperation or temptations of defecting also vary. Because bonuses are paid, a policy of giving bonuses for cooperative behavior imposes an actual cost on an organization. In contrast, whistleblower rewards deter collusion and are not dispersed on the equilibrium path.

\paragraph{The Relationship to the Literature.} 

Game theory, team theory, and mechanism design generally take the design problem as designing incentives for one particular (in some cases, dynamic) context.  Here, we examine both the design of who faces which dynamic sequence of contexts, as well as with whom they play.%
\footnote{There is a literature that examines how the diversity of teams affects their performance potential (e.g., see \cite{page2008difference}), but that is based on team heterogeneity, which is not considered here, and that literature does not consider any of the incentive issues or dynamics considered here.} To our knowledge, this is the first paper to analyze the interaction of incentives across games, as well as the dynamic design of the sequence of games faced.  This also provides a first analysis of treating the assignments of players to teams, and teams to games, as endogenous.  This provides new angles for understanding team and organizational design.   

There is a small literature concerned with cooperation across multiple games \citep{bardhan1980, karnaniw85,rotemberg1986supergame,bernheimw1990,watson1999,watson2002,spagnolo1991,rauchw2003}. Relative to existing work, ours is the first to analyze interactions across games in which cooperation is desired in some but not others, as well as team assignment, and the implications for organizational design. 

There is extensive study of cooperation and the patterns of bilateral relationships  \citep{kandori1992,ellison1994,kranton1996,goshr1996}, the incentives to oust dissidents \citep{alim2016}, monitoring technologies \citep{wolitzky2012,navap2012}, community network structure \citep{lipperts2011,jacksonrx2012,alim2013,balmacedae2017}, and what types of punishment are preferred \citep{acemogluw2019}. We intentionally abstract from many of these issues, instead focusing on how the organizational structure affects cooperation.

Our model is a stochastic game (introduced by \citealp{shapley1953}; see \citealp{solanv2015} for a short review). Our model's special structure---the players' current actions do not affect the distribution over subsequent games---may be useful for extensions. This environment allows us to characterize equilibria, which are typically difficult to solve for \citep{solanv2015}.  Moreover, our analysis involves the design of a stochastic game, rather than analyzing a given one.

A number of studies analyze corruption as a by-product of institutional and organizational structures  (e.g., \cite{basubm92,shleiferv93,myerson1993, bardhanm2000, acemoglur19}). Our analysis contributes by showing how corruption depends not only on the benefits and costs of corrupt collusion but also on how parties interact in other domains. Furthermore, our analysis provides insights into why corruption may be tolerated in some countries and organizations and why eliminating it requires wide-ranging interventions and can have unintended consequences.

Some studies have asked how ``walls of silence,'' where group members conceal information that may harm other members, are sustained \citep{benoitd2004,muelheusserr2008}. In \cite{benoitd2004}, honest agents may vote for a regime that conceals misconduct in fear of a type 2 error in which they are innocent but found guilty. \cite{muelheusserr2008} show how reputational concerns and the threatened loss of future cooperation can induce agents to remain silent. Our model differs: agents are homogeneous, and there is no private information. Our comparative statics also differ: we consider how cooperation relates to organizational design.

\section{The Model}\label{Model}

We model a team's strategic interactions as a stochastic game played by $n$ players. The ensuing model allows us to characterize a team's equilibrium cooperation as a function of the team's durability and the tasks it is assigned to; we enrich the model to consider the dynamic assignment of teams to tasks at a later stage.

Time proceeds discretely, $\tau=0,\dots,\infty$. The stage games, or simply games for brevity, played at those times are of two varieties, $\{G,B\}$, with generic element $g$. As we formalize later, the notation $G,B$ denotes that the first type of game involves cooperation that is socially good; the second, cooperation that is socially bad. Stage game $g$ occurs with probability $p_g>0$ such that $p_G+p_B\leq 1$, and at most one game occurs. For example, police patrols might be told to respond to crimes that occur at random times. Players discount payoffs received at time $\tau$ by $\delta^\tau$, where the discount factor $\delta \in (0,1)$.

Cooperation in the stage games produces overall benefits to the players, but there are also temptations to deviate from cooperating. For simplicity, we assume binary actions: in each game $g$, a player $i$ chooses either to `cooperate' or `not': $A=\{C,N\}$. 

To simplify the analysis, we impose a number of restrictions on the stage games. 

\begin{properties}Let $\pi_{i,g}(\mathbf{a})$ be the payoff to player $i$ in game $g$ as a function of the action profile $\mathbf{a}\in A^n$.
\begin{enumerate}
    \item \label{property:symmetry}For all players $i$, action profiles $\mathbf{a}=(a_i,\mathbf{a}_{-i})$ and $\mathbf{a}_{-i}'$, where $\mathbf{a}_{-i}'$ is a permutation of $\mathbf{a}_{-i}$, we have 
    $$\pi_{i,g}(\mathbf{a})=\pi_g(a_i,\mathbf{a}_{-i})=\pi_g(a_i,\mathbf{a}^\prime_{-i}),$$
    for some function $\pi_g$,
    \item \label{property:max} aggregate payoffs are uniquely maximized when all players cooperate, i.e., $$C^n=\argmax_{\mathbf{a}\in A^n}\sum_i \pi_g(a_i,\mathbf{a}_{-i}),$$
    \item \label{property:dom} not cooperating is a strictly dominant action, i.e.,
    $$
        \pi_g(N,\mathbf{a}_{-i})>\pi_g(C,\mathbf{a}_{-i}) \text{ for all }\mathbf{a}_{-i} \in A^{n-1},
    $$
    \item \label{property:minmax} a player's payoff is minimized by complete non-cooperation of others, i.e.,
    $$
        N^{n-1}\in\argmin_{\mathbf{a_{-i}}\in A^{n-1}}\left\{\max_{a_i \in A}\pi_g (a_i,\mathbf{a_{-i}})\right\},
    $$
    \item\label{property:last} aggregate payoffs are strictly larger with complete non-cooperation than with any asymmetric play, i.e., 
    $$
        N^n= \argmax_{\mathbf{a} \in A^n\setminus\{ C^n\}}\sum_i \pi_g(a_i,\mathbf{a_{-i}}),
    $$
    \item \label{property:star} and the gain from playing $N$ instead of $C$ is smallest when all other players play $C$, i.e.,  $$\pi_g(N,C^{n-1})-\pi_g(C,C^{n-1})\leq \pi_g(N,\mathbf{a_{-i}})-\pi_g(C,\mathbf{a_{-i}})\text{ for all }\mathbf{a_{-i}} \in A^{n-1}.$$    
\end{enumerate}
\end{properties}

Properties~\ref{property:symmetry}--\ref{property:dom} are standard: Property~\ref{property:symmetry} states that a player's payoff is only a function of their own action and the aggregate actions, not their own identity and the identities of others, Property~\ref{property:max} implies that players benefit from cooperation, while Property~\ref{property:dom} means each player individually benefits from not cooperating in a single game. Together with Property~\ref{property:minmax}, Property~\ref{property:dom} also implies a lower bound on the payoffs a player receives in any continuation, as well as an equilibrium that achieves them. Lastly, Properties~\ref{property:last} and~\ref{property:star} simplify the analysis by ensuring that, on the equilibrium path, asymmetric play is never player-optimal, i.e., in any particular game, there is either full cooperation or full non-cooperation. As we show in Lemma~\ref{lem:equilibrium_characterization}, these properties allow us to characterize the games using just a few parameters described below. Thus, we can focus our attention on interactions across games, rather than the details of dynamics within a game, which can be quite complex absent such structure.\footnote{For example, see \cite{olszewskis2018}. Beyond the dynamic differences, their game also involves private information.}

The incentives in game $g$ are captured by the following parameters. Complete cooperation generates a benefit, $c_g$, relative to complete non-cooperation for all players: $c_g\equiv {\pi}_g(C,C^{n-1})-{\pi}_g(N,N^{n-1})$. There is a deviation temptation, $d_g$, which captures the gain in stage payoff to any agent who deviates when the others cooperate: $d_g\equiv {\pi}_g(N,C^{n-1})-{\pi}_g(C,C^{n-1})$. Hence, in our base model, the stochastic game is described by the number of players, $n$, the discount factor, $\delta$, and the cooperation benefit $c_g$, deviation temptation $d_g$, and likelihood $p_g$ of each stage game $g$. 

For example, in a prisoner's dilemma game $g$, there are $n=2$ players with payoffs: 
$$
    \centering
    \begin{tabular}{cc|c|c|}
      & \multicolumn{1}{c}{} & \multicolumn{2}{c}{Player 2}\\
      & \multicolumn{1}{c}{} & \multicolumn{1}{c}{C}  & \multicolumn{1}{c}{N} \\\cline{3-4}
      \multirow{2}*
      {Player 1}  & C & $(c,c)$ & $(-a,c+d)$ \\\cline{3-4}
      & N & $(c+d,-a)$ & $(0,0)$ \\\cline{3-4}
    \end{tabular}
    \smallskip
$$
where $c,d>0$ and $a> c+d$. All players cooperating maximizes the total payoff, but a player can benefit by unilaterally deviating. The condition that $a> c+d$ ensures that the game satisfies Properties~\ref{property:last} (since then $0>c+d -a$) and~\ref{property:star} (since then $c+d - c\leq 0 - (-a)$). The resulting cooperation benefit and deviation temptation are $c_g\equiv c$ and $d_g\equiv d$, respectively.

Other games, such as a favor exchange game in which players randomly have chances to either give or receive favors---both good and bad (e.g., lending money vs.\ not reporting a friend's crime)---also satisfy these properties. Different players move at different points in time, but a similar analysis applies.

Players condition their actions considering the full stochastic game. The history of the stochastic game at time $\tau$ encodes all actions taken up to time $\tau$, i.e., it is an element $\{(g_{\tau'},\mathbf{a}_{\tau'})\}_{\tau'=0}^{\tau-1}$, where $g_{\tau'},\mathbf{a}_{\tau'}$ are the game (which may be no game) and the action profile played at time $\tau'$. Let $\mathcal{H}_\tau$ denote all such histories and $\mathcal{H}\equiv \cup_\tau \mathcal{H}_\tau$. A strategy for a player is then a mapping $\sigma: \mathcal{H} \times \{G,B\} \rightarrow \Delta(C,N)$. With this setup, we consider subgame-perfect equilibria. 

\subsection{Cooperation}\label{sec:coop}
We say that an equilibrium strategy profile has \emph{cooperation in game $g$} if, on the equilibrium path, all players cooperate whenever game $g$ is played. We say there is \emph{no cooperation in game $g$} if, on the equilibrium path, all players do not cooperate whenever game $g$ is played. This leads to four types of canonical cooperation: 
$$
\centering
    \begin{tabular}{c|c|c|}
       \multicolumn{1}{c}{} & \multicolumn{1}{c}{$C^n$ in $B$}  & \multicolumn{1}{c}{$N^n$  in $B$} \\\cline{2-3}
       $C^n$  in $G$ & $\emph{total cooperation}$ & \emph{cooperation only in $G$} \\\cline{2-3}
       $N^n$  in $G$ & \emph{cooperation only in $B$} &  \emph{no cooperation} \\\cline{2-3}
    \end{tabular}
$$

The different kinds of cooperation can be partially ordered by some social criterion. Let $V_G>0$ denote the social benefit from cooperation in game $G$, the good game, from the perspective of a social planner or designer.  For example, this could be the social benefits to a society's broader population of having a military or police force that cooperates when in dangerous situations, or the profits to a company from having a team cooperate on a productive task, etc. Similarly, $V_B<0$ denotes the social harm from cooperation in the bad game $B$. Thus, cooperation only in game $G$ is clearly most preferred; we call this \emph{optimal cooperation}. Cooperation only in game $B$ is clearly the worst. The social value of total cooperation, $p_G V_G+p_B V_B$, and the social value of no cooperation, $0$, lie in between the best and worst outcomes; their ordering depends on the relative social impact of cooperation in the two games and on their probabilities.

In what follows, we sometimes refer to organizational design and a designer. However, we remain agnostic about who designs the organization or whether they are benevolent. In some situations, these are decision-makers, such as a police chief or a firm's executives, who can purposefully implement an organizational design such as reshuffling teams. Sometimes, however, there is no explicit designer; in such circumstances, our results clarify whether it is possible to reach optimal cooperation as a function of the environment. 

\subsection{Player-Optimal Equilibria}\label{sec:equilibrium_characterization}
As there can be multiple equilibria (the usual folk-theorem arguments apply), we focus on those that are optimal for the players. Given the game's symmetry and its required properties, we characterize the equilibria that maximize the sum of players' expected discounted payoffs.

Characterizing such player-optimal equilibria makes sense for three reasons. First, they are extremal---supporting maximal cooperation---and thus define the frontier and thus the full set of equilibrium payoffs, regardless of what one believes about equilibrium selection. Second, it may not be robust to presume that a designer can guide players to one equilibrium if all players prefer another.\footnote{Our equilibrium selection coincides with that in \cite{rotemberg1986supergame}, who focus on the most collusive symmetric equilibrium in a stochastic repeated game.} Third, the possibility of interaction across good and bad games that is highlighted here is the object of interest, and (Pareto-dominated) equilibria that involve no such interaction could also be analyzed, but involve no new insights or intuitions. More generally, once one departs from the player-optimal equilibrium, equilibrium behavior need not admit simple characterizations, as we show below, and the designer's problem no longer reduces to transparent incentive constraints. Analyzing arbitrary equilibria would, therefore, require additional structure or selection assumptions and would obscure the mechanism that is central to our analysis.

\begin{lemma}\label{lem:equilibrium_characterization}
In any player-optimal subgame-perfect equilibrium, on the equilibrium path, there is total cooperation if \begin{equation}\label{eq:full_cooperation}
     \max_g\{d_g\}\leq \sum_g\frac{\delta }{1-\delta}p_gc_g;
\end{equation}
there is cooperation only in game $g$ if 
\begin{equation}\label{eq:partial_cooperation}
    d_g\leq \frac{\delta }{1-\delta}p_gc_g
\end{equation}
and \eqref{eq:full_cooperation} does not hold; and otherwise, there is no cooperation.
\end{lemma}

All proofs are in Appendix~\ref{appendix:proofs}.

The arguments of the proof follow standard repeated-game logic. When Inequality~\eqref{eq:full_cooperation} holds, the grim-trigger strategy `cooperate in all games unless some player did not cooperate in the past, at which point never cooperate in any game again' constitutes an equilibrium: no player has an incentive to deviate in either game, and players play the stage game Nash equilibria during the punishment stage. Furthermore, this equilibrium is player-optimal as full cooperation maximizes aggregate stage-game payoffs. Similarly, when \eqref{eq:full_cooperation} does not hold, but \eqref{eq:partial_cooperation} does for game $g$, then `cooperate in $g$ only unless any player deviated from doing so in the past, in which case do not cooperate in either game' constitutes an equilibrium: since \eqref{eq:full_cooperation} fails, players cannot cooperate in both games as otherwise they would deviate, and so their payoffs cannot be improved, now also leveraging that no cooperation maximizes aggregate stage-game payoffs excluding full cooperation. Lastly, since any equilibrium-path cooperation in both games requires \eqref{eq:full_cooperation}, and cooperation in game $g$ requires \eqref{eq:partial_cooperation}, when both conditions fail, players maximize payoffs by playing the stage Nash equilibrium. 

Visually depicting the different cases of equilibrium cooperation characterized in Lemma~\ref{lem:equilibrium_characterization} is useful. Figure~\ref{fig:basic_cases} shows which case occurs as a function of the probabilities of good, resp.\ bad, games occurring times $\delta/(1-\delta)$. Taking Inequality~\eqref{eq:full_cooperation} as an equality delineates the region where total cooperation prevails. Cooperation only in some game $g$ precludes total cooperation, so the former region must be southwest of the latter. For cooperation to be sustainable, Inequality~\eqref{eq:partial_cooperation} must hold; its equality delineates the region with (some) cooperation. Note: there cannot be cooperation only in game $g$ if the temptation to deviate in that game exceeds the deviation temptation in the other game; thus, only one such region exists. 

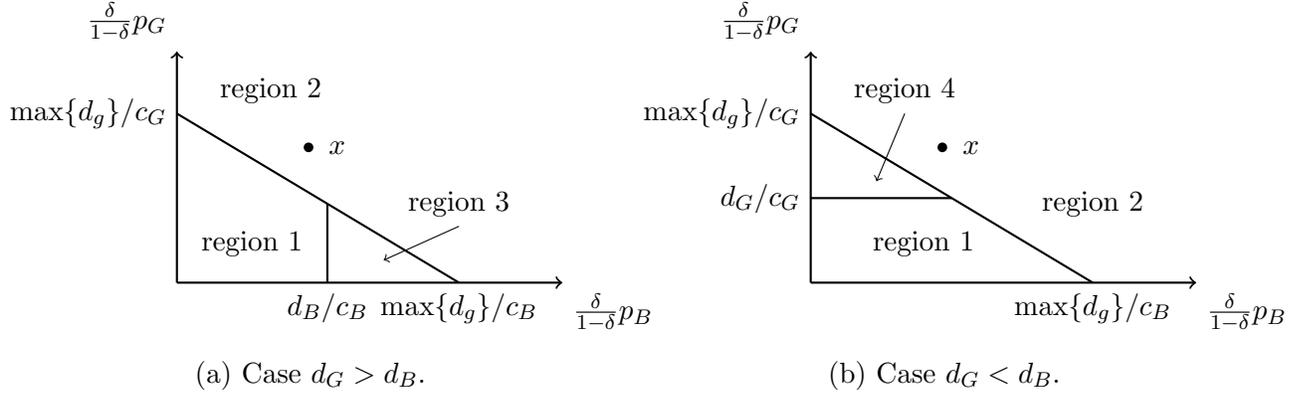
\begin{figure}[ht]
\small
\captionsetup{width=0.9\textwidth}
    \centering
    \begin{subfigure}[t]{0.5\textwidth}
        \centering
        \begin{tikzpicture}
        \def \xscale {.25} 
        \def \yscale {.15}
            \draw[thick,->] (0*\xscale,0*\yscale) -- (20.5*\xscale,0*\yscale) node[anchor=north west] {$\frac{\delta}{1-\delta}p_B$};
            \draw[thick,->] (0*\xscale,0*\yscale) -- (0*\xscale,20.5*\yscale) node[anchor=south east] {$\frac{\delta}{1-\delta}p_G$};
            \draw[thick] (0*\xscale,15*\yscale) node [anchor=east] {$\max\{{d_g}\}/c_G$} --(15*\xscale,0*\yscale) node [anchor=north] {$\max\{{d_g}\}/c_B$};
            \draw[thick] (8*\xscale,7*\yscale) -- (8*\xscale,0*\yscale) node [anchor=north] {$d_B/c_B$};
            \draw (5*\xscale,15*\yscale) node [anchor=south] {region 2};
            \draw (4*\xscale,3.5*\yscale) node[align=center] {region 1};
            \draw (15*\xscale,5*\yscale) node[anchor=south,align=center] {region 3};
            \draw[->] (15*\xscale,5*\yscale) -- (11*\xscale,2*\yscale);
        \end{tikzpicture}
        \caption{Case $d_G>d_B$.}
        \label{subfig:bad}
    \end{subfigure}%
    ~ 
    \begin{subfigure}[t]{0.5\textwidth}
        \centering
        \begin{tikzpicture}
        \def \xscale {.25} 
        \def \yscale {.15}
            \draw[thick,->] (0*\xscale,0*\yscale) -- (20.5*\xscale,0*\yscale) node[anchor=north west] {$\frac{\delta}{1-\delta}p_B$};
            \draw[thick,->] (0*\xscale,0*\yscale) -- (0*\xscale,20.5*\yscale) node[anchor=south east] {$\frac{\delta}{1-\delta}p_G$};
           \draw[thick] (0*\xscale,15*\yscale) node [anchor=east] {$\max\{{d_g}\}/c_G$} --(15*\xscale,0*\yscale) node [anchor=north] {$\max\{{d_g}\}/c_B$};
            \draw[thick] (7.5*\xscale,7.5*\yscale) -- (0*\xscale,7.5*\yscale) node [anchor=east] {$d_G/c_G$};
            \draw (15*\xscale,5*\yscale) node[anchor=south, align=center] {region 2};
            \draw (6*\xscale,3.5*\yscale) node[align=center] {region 1};
            \draw (5*\xscale,15*\yscale) node[anchor=south,align=center] {region 4};
            \draw[->] (5*\xscale,15*\yscale) -- (3.5*\xscale,9*\yscale);
        \end{tikzpicture}
        \caption{Case $d_G<d_B$.}
        \label{subfig:good}
    \end{subfigure}
    \caption{Different types of cooperation. In region~1, there is no cooperation.  In region~2, there is total cooperation. In region~3, there is cooperation in the bad game only.  In region~4, there is cooperation in the good game only.}
    \label{fig:basic_cases}
\end{figure}

Although the games are independent, their strategic interactions necessitate a holistic analysis to understand players' actions in one of them. For example, myopically improving the incentives for desirable actions in one type of situation may backfire, generating detrimental behaviors elsewhere. Similarly, collusive behavior may be minimized by eradicating situations where it occurs (decreasing $p_B$), reducing its benefits (decreasing $c_B$), or incentivizing team members to stop colluding or to report it (increasing $d_B$); although all these may serve their intended purpose of reducing collusion, they also can impair productive cooperation in other situations if incentives elsewhere are not adjusted to compensate. The dual of this warning highlights missed opportunities: A myopic design approach may overlook opportunities to reduce bad cooperation by overlooking opportunities to reduce the need for good cooperation. We discuss the implications of changing payoffs in greater detail in the online appendix.

\section{How Cooperation Depends on Team Durability}
\label{sec:designers_problem}

Given the characterization of equilibria, we can now examine how cooperation depends on different dimensions of organizational design. We begin with team durability. A team's durability is the (expected) time it stays together before teams are reshuffled and new ones formed. If team durability is low, i.e., teams are frequently reshuffled (at random times), future benefits from cooperation with current team members are reduced, and so teams become less likely to cooperate.\footnote{Reshuffling at random times will help sustain cooperation, as deterministic reshuffling has incentives unravel from the final date at which a team knows it will be together and thus cannot sustain any cooperation.} 

We model the reshuffling of teams (reducing team durability) as a probability $r$ with which the stochastic game of a team ends after each period. Given the continuum of workers, we assume that reshuffled players never face a past teammate, that there is no spreading of information, and that players only condition on histories with their current team.\footnote{This greatly simplifies the analysis; see, e.g., \cite{alim2016} for a treatment where such contagion is possible. Our simplification also requires sufficiently many workers so that completely new teams can be formed, as well as anonymity of workers outside their teams (deviations can be punished only by current team members). These conditions are more likely to hold---at least approximately---in large organizations, potentially leading to differences by organization size that we leave for future research. This can also be handled by hiring people part-time or by outsourcing and subcontracting tasks on short-term rotations.} Consequently, a reshuffling rate $r$ is equivalent to changing the discount factor $\delta$ to $(1-r)\delta$; essentially, the designer can induce any effective discount factor of at most $\delta$ by varying team durability. Thus, formally, we characterize optimal reshuffling by examining the player optimal equilibria for a given team (with histories only applying to that team) as a function of discount factor $(1-r)\delta$.

Our first result characterizes when reshuffling can lead to optimal cooperation. 

\begin{proposition}
\label{prop:optimal_reshuffling}
There exists a reshuffling probability $r$ such that there is optimal cooperation if and only if all of the following hold:
\begin{enumerate}
    \item The temptation to deviate is less in the good game: $d_G<d_B$;\label{cond:deviation}
    \item cooperation is sustainable in the good game by itself: $d_G\leq \frac{\delta}{1-\delta} p_Gc_G$;\label{cond:absolute}
    \item the good game's likelihood is high compared to that of the bad game: $\frac{p_Gc_G}{p_Bc_B}>\frac{d_G}{d_B-d_G}$.\label{cond:relative} 
\end{enumerate}
\end{proposition}

The temptation to deviate in good games must be no higher than that in bad games;  otherwise, whenever good cooperation is possible, then so is bad. Thus, condition~\ref{cond:deviation} is necessary. Cooperation being sustainable in the good game by itself (condition \ref{cond:absolute}) is also necessary because, for optimal cooperation, there are no strategic interdependencies across the games. Lastly, good games must have a sufficiently high probability, relative to bad games (condition~\ref{cond:relative}), to make it feasible (possibly by reshuffling) to induce players to cooperate in the good game but not in the bad one: it ensures that there exists an effective discount rate $(1-r)\delta$ such that Equation~\eqref{eq:full_cooperation} fails, while \eqref{eq:partial_cooperation} holds for the good game. The proposition then follows from a careful application of Lemma~\ref{lem:equilibrium_characterization}.

When any of the above conditions fail, there is a trade-off between two evils: accept corruption in bad games or face cooperation in neither game. If total cooperation is feasible then it is preferred to no cooperation if and only if $p_GV_G+p_BV_B>0$. In policing, for example, productive team cooperation is vital to society---$V_G$ is high. However, because police misconduct imposes significant harm, $V_B$ can also be large \citep{mollen1994}. Thus, whether frequent rotations or long-term partnerships are optimal depends on the relative frequency of situations in which (a) police need to back each other up or perform other important team duties versus (b) they can be corrupt or abusive.

In other settings---e.g., low-skilled jobs in retail, fast food, and warehouse industries---cooperation may increase profits only slightly ($p_GV_G$ is small), whereas employee collusion to steal products is a major cost for retailers ($p_BV_B$ is large in absolute terms; \cite{greenberg1996employee}). This helps to explain why retail teams, unlike police or military units, are reshuffled frequently. Lastly, consider sports teams and academics. While both professions have some potential for collusion \citep{muelheusserr2008}, the social cost may be outweighed by the value of cooperation, and so teamwork is often encouraged via design.

\section{Assigning Teams to Tasks}\label{sec:task_assignment}

In the previous sections, the arrival rates of good and bad games were fixed, and the only design dimension was the rate at which teams were reshuffled.  We now consider situations in which there are different tasks to which teams can be assigned, and the arrival rates of the two games differ across these tasks, thereby allowing an organizational designer to control arrival rates. For example, in a police force, the task of responding to emergencies might create infrequent but important situations in which backup is required (good games), whereas the task of traffic policing might present frequent opportunities to accept small bribes (bad games).

As a running illustration of optimal assignments in this and the next section, consider the following example.   

\begin{example}\label{example}
Payoffs are such that $\delta=3/5,c_G=c_B=1,d_G=1$, and $d_B=1/2$, and the arrival probabilities of games are $p_G=1/4$ and $p_B=3/8$.

In this setting, cooperation on good games is not sustainable. To see this, note that \eqref{eq:full_cooperation} is violated. That is, the payoff from deviating in a good game and not cooperating exceeds the maximum potential future payoff loss, presuming that there could be full cooperation on both games:
$$
d_G=1>\frac{15}{16}=\frac{\frac{3}{5}}{1-\frac{3}{5}}\cdot\left( \frac{1}{4}\cdot 1 + \frac{3}{8}\cdot 1 \right)=\frac{\delta}{1-\delta}(p_Gc_G+p_Bc_B).
$$

Suppose, however, that the relative arrival rates of good and bad games can be adjusted. In particular, there are two tasks that have to be covered with equal weights---at any point in time, half of the teams need to cover each---and on one task, (only) good games arrive at a rate of $a_G=1/2$, and on the other task (only) bad games arrive at a rate of $a_B=3/4$. If a team is randomly assigned to each task in each period with equal probability, then this satisfies the task coverage constraint and gives back $p_G=1/4=\frac{1}{2} a_G$ and $p_B=3/8=\frac{1}{2}a_B$ as the arrival probabilities of future games. 

But now, instead of assigning a team to each task with equal probability, one could assign it to different mixtures of tasks. Ideally, assigning it entirely to the task with good games would sustain cooperation. Then, other teams could be assigned to the task with bad games (so that each task is covered) and reshuffled to preclude cooperation, leading to optimal cooperation.

However, note that in this example, if one assigns a team entirely to the task with only good games, that would also not lead to cooperation on good games since future good games alone do not provide enough incentives to cooperate as the incentives to deviate exceed any benefits from future cooperation:
$$
d_G=1>\frac{3}{4}=\frac{\frac{3}{5}}{1-\frac{3}{5}}\cdot \frac{1}{2}\cdot 1=\frac{\delta}{1-\delta}a_Gc_G
$$
(a condition which we later formalize in Observation~\ref{obvs:full_specialization}).

Thus, cooperation on good games cannot be sustained with the original default arrival rates, nor with assignment to just the task with only good games. Instead, what is necessary to sustain cooperation is an assignment that includes a sufficient stream of future \emph{bad} games that then provide rewards and incentives not to deviate in a current good game. Indeed, there is cooperation on both games when teams face a high enough share of future bad games since $\max_g\{d_g\} = 1 < \frac{9}{8}= \frac{\frac{3}{5}}{1-\frac{3}{5}}\cdot \frac{3}{4}\cdot 1 = \frac{\delta}{1-\delta} a_Bc_B$. 

To illustrate, consider a task assignment in which a team is assigned to the task with only good games with probability $\nu_G$ and to the task with only bad games with probability $\nu_B=1-\nu_G$.  The assignment that provides enough incentives to sustain cooperation in good games, while minimizing the share of bad games, solves
\begin{equation}\label{eq:example_coop}
\max\left\{\frac{1}{2},1\right\}=\frac{\frac{3}{5}}{1-\frac{3}{5}}\left(\nu_G\cdot \frac{1}{2} \cdot 1+\nu_B\cdot \frac{3}{4}\cdot 1\right),
\end{equation}
and is thus given by $\nu_G=1/3$ and $\nu_B=2/3$. 
\end{example}

What Example~\ref{example} shows is that if the organizational designer has additional flexibility to steer a team towards tasks that involve more good or bad games, then they can find mixes that sustain cooperation with minimal arrivals of bad games.  

The example illustrates just one part of the overall design problem.  The designer still has to cover all tasks with teams, that is, in the previous example, achieve the 50--50 assignment of teams to tasks on average.  So, by tilting some teams toward this optimal mix, some tasks will not be fully covered.  Those can be covered with other teams, but ones in which the designer simply gives up on motivating cooperation. As a result, some teams end up incentivized to cooperate, while others are left without cooperation---the assignment will often be asymmetric.  Exactly which assignment ends up being optimal depends on the setting: the relative arrival rates, payoffs, and the required task coverage.  

We model this additional design option as follows. There is a finite set of tasks $\mathcal{T}$. Each task $t$ has its own probability of good and bad games, $p_G(t)$ and $p_B(t)$. An \emph{assignment} for a team, denoted $\nu$, is a probability measure on $\mathcal{T}$ that specifies the probability with which the team is assigned to each task, with the realized task redrawn in every period. For example, if there are two tasks, responding to emergencies vs.\ traffic policing, then $\nu$ indicates the probability that the team is assigned to respond to emergencies as well as the (remaining) probability that it is assigned to traffic policing. 

An organization's \emph{team structure} is a probability measure $\lambda$ over assignments (so  $\lambda\in \Delta\left(\Delta(\mathcal{T})\right)$), together with a reshuffling rate $r_\nu$ for each $\nu \in \supp(\lambda)$. Thus, a team structure indicates the fraction of the population that gets each type of assignment. 

A team's history at time $\tau$ is now augmented with the history of assigned tasks, i.e., it is a vector $[(t_{\tau'},g_{\tau'},\mathbf{a}_{\tau'})]^{\tau-1}_{\tau'=0}$ where $t_{\tau'}$ is the task faced at time $\tau'$. Let $\mathcal{H}_\tau$ denote all such histories and $\mathcal{H}\equiv \cup_\tau \mathcal{H}_\tau$. A player's strategy indicates what the player who falls under a particular $\nu$ does as a function of the history, i.e., $\sigma:\mathcal{H}\times \{G,B\} \rightarrow \Delta(C,N)$, similar to before.  Again, we track histories only within a team and consider equilibria that depend only on current team histories.

In our analysis, only a finite number of assignments are needed, and so we treat $\lambda$ as having finite support. With a continuum of agents and a finite set of assignments, agents can be reshuffled, as in the previous section, to a new team that is given the same assignment.  Again, to avoid technical clutter, we do not model the reshuffling explicitly as a continuum-matching process but simply capture it via a reduced discount rate.

The team structure must respect an exogenous constraint that specifies the fraction of the organization's teams assigned to a given task at any given instant, represented by $q_t$, where $\sum_{t\in \mathcal{T}} q_t=1$. In other words, the organization's team structure $\lambda$ must satisfy
\begin{equation}\label{eq:balance}
    \sum_{\nu \in \supp(\lambda)}\lambda(\nu) \nu(t)  =  q_t \quad \text{for all }t\in \mathcal{T}.
\end{equation}

For reactive assignment $\nu$, the average arrival probabilities of the different games are given by  $p_g(\nu) \coloneqq \sum_{t\in\mathcal{T}} \nu(t)\, p_g(t)$. For $p_G(\nu),p_B(\nu)>0$, equilibrium play is still characterized Lemma~\ref{lem:equilibrium_characterization}, that is, there are still the four types of canonical cooperation, and, of course, when only one game has a positive probability of arrival, then cooperation is sustained if and only if $d_G\leq \frac{\delta}{1-\delta}p_g(\nu)c_g$. This implies a well-defined social value. For any team structure $\lambda$, together with reshuffling rate $r_\nu$ for each $\nu \in \supp(\lambda)$, let $\chi_g(\nu)\in\{0,1\}$ denote the
equilibrium indicator of cooperation in game $g\in\{G,B\}$ under assignment
$\nu$. 
The social value generated by team structure $\lambda$ is then
$$
\sum_{\nu\in\supp(\lambda)} \lambda(\nu)
\sum_{g\in\{G,B\}}
p_g(\nu)\,\chi_g(\nu)\,V_g.
$$
This can be taken as an objective function for the designer in choosing $\lambda$.

This model nests the model from our previous sections, where there was just one task $t$ with $q_t=1$, all teams were assigned to $\nu$ with $\nu(t)=1$, and the designer's only decision was the reshuffling rate $r_\nu$. We are now considering the more general case in which there are several tasks and discretion over how teams are assigned to those tasks.\footnote{The previous model is also nested if interpreted as there being several tasks, but the designer has no discretion over task assignments; i.e., they must use the same assignment for each team, namely $\nu$ with $\nu(t)=q_t$.}

The essential insights are easy to see in the case of just two tasks, as in Example~\ref{example}, and so we work with that case. Specifically, let $\mathcal{T}=\{t_G,t_B\}$ with corresponding fractions $q_G,q_B$ of teams that must cover each task (we write $q_G,q_B$ instead of $q_{t_G},q_{t_B}$ for ease of notation). For simplicity, we further assume that task $t_G$ has an arrival probability $a_G= p_G(t_G)>0$ of the good game $G$ and no probability of $B$, and correspondingly, task $t_B$ has an arrival probability $a_B= p_B(t_B)>0$ of the bad game $B$ and no probability of $G$.\footnote{Nothing substantive changes with tasks that are not extremal; i.e., when they generate both good and bad games. 
The designer has less flexibility in shaping the stream of games faced by a team, since assignments cannot fully separate the two types.  Similarly, introducing more than two tasks with various heterogeneous arrival rates complicates the full characterization of the optimal assignment, but not the basic forces that we analyze of trading off more social and less socially desirable arrival rates to generate cooperation with minimal relative amounts of cooperation in bad games.} We therefore refer to tasks $t_G$ and $t_B$ as ``good'' and ``bad'' tasks, respectively. This is not to say that completing the bad task is harmful to the organization; in fact, a share $q_B$ of teams must be covering that task.  The ``bad'' simply refers to the fact that cooperation, compared to noncooperation on that task, is worse for society. For example, even though traffic policing might present frequent opportunities to accept bribes, it is valuable to society and must be done, so there is an exogenous organizational constraint specifying the share of the population that must be assigned to it. 

An organization with only generalists---i.e., where every team faces the same set of tasks---is captured by a single reactive assignment $\nu^{gen}$ with $\nu^{gen}(t_G)=q_G$ and $\nu^{gen}(t_B)=q_B$. Another team structure is full specialization, with one assignment being to only the good task, $\nu^G(t_G)=1$, and another assignment being to only the bad task, $\nu^B(t_B)=1$ (task coverage then requires $\lambda(\nu^G)=q_G$ and $\lambda(\nu^B)=q_B$). In this case, the designer optimally keeps teams on the good task together ($r_G=0$; again, we write $r_G$ instead of $r_{\nu^G}$ for simplicity, as we do for $r_B$) and reshuffles those on the bad task ($r_B=1$). This specialization of teams overcomes two of the three conditions required for optimal cooperation from Proposition~\ref{prop:optimal_reshuffling} (conditions~\ref{cond:deviation} and~\ref{cond:relative}): First, greater temptation to deviate in good than bad games still means that desirable cooperation leads to collusion; however, the only teams that encounter bad games are frequently reshuffled and hence never cooperate. Second, good and bad games are fully separated; hence, there are no spillovers constraining the relative frequencies of good and bad games. Also, because cooperation becomes easier when interactions occur more frequently, the condition for optimal cooperation with full specialization on good games is weaker than its analog (condition~\ref{cond:absolute}). Thus, full specialization is an optimal team structure if cooperation can be sustained on the good task by itself. This is easy to see, and so we simply state it as an observation.

\begin{observation}\label{obvs:full_specialization}
Full specialization, together with frequent enough reshuffling of those teams assigned to the bad task and no reshuffling of teams assigned to the good task, leads to cooperation only on good games if and only if $d_G\leq \frac{\delta}{1-\delta}a_G c_G$.
\end{observation}

The more interesting case is one where the full-specialization team structure does not yield such optimal cooperation; therefore, we investigate situations in which sustaining cooperation in the good game can only be done if there is a high-enough probability of future bad games ($d_G> \frac{\delta}{1-\delta}a_G c_G$ and $\max_g \{d_g\} < \frac{\delta}{1-\delta} a_B c_B$). In these situations, the optimal team structure may involve creating two types of team assignments: one that includes the bad task with the minimum share needed to sustain total cooperation (as in \eqref{eq:example_coop} in Example~\ref{example}), and another that ensures all tasks are covered according to the exogenous constraint given in \eqref{eq:balance}.

Thus, consider assignment $\nu^{coop}$, which generalizes \eqref{eq:example_coop} in Example~\ref{example}, that places just enough weight on the bad task to incentivize cooperation in both games, i.e.,
\begin{equation}
\label{hybrid}
\max_g \{d_g\} = \frac{\delta}{1-\delta} \left(\nu^{coop}(t_G) a_G c_G + (1-\nu^{coop}(t_G)) a_B c_B\right).
\end{equation}

Either assignment $\nu^{coop}$ is part of an optimal team structure, or the social cost of bad cooperation is too large and no cooperation is optimal. In the former case, the share of the population with assignment $\nu^{coop}$ is maximized, so that, under the assignment, teams face only good or only bad tasks and are constantly reshuffled to preclude cooperation. The next proposition formalizes this. 

\begin{proposition}\label{prop:task_assignment}
Suppose that $d_G> \frac{\delta}{1-\delta}a_G c_G$ and $\max_g \{d_g\} < \frac{\delta}{1-\delta} a_B c_B$. Then either (i) the optimal team structure $\lambda$ has positive mass on at most two assignments, namely
\begin{itemize}
    \item $\nu^{coop}$ (defined in Equation \ref{hybrid}) with $r_{\nu^{coop}}=0$ and
    \item $\nu^{g}$ with $\nu^{g}(t_{g})=1$, for some $g\in \{G,B\}$, and $r_g$ sufficiently high to preclude cooperation if $g=B$,\footnote{If $g=G$, then there is no cooperation for any reshuffling rate since we are considering settings in which bad tasks are needed to sustain cooperation in good games.}
\end{itemize}
or (ii) any team structure having only assignments that lead to no cooperation is optimal.\footnote{Note that the relative weights on the two assignments in case (i) are fully pinned down: the optimal team structure $\lambda$ satisfies $\lambda(\nu^{coop}) \nu^{coop}(t_{g'}) = q_{g'}$  and $\lambda(\nu^{g})+\lambda(\nu^{coop})\nu^{coop}(t_{g})=q_{g}$, where $g' \in \{G,B\}$ and $g'\neq g$.}
\end{proposition}

Let us revisit Example~\ref{example} and explicitly derive the optimal team structure for these parameters.

\begin{continuance}{example}
A team structure $\lambda$ with positive weight on $\nu^{coop}$ can be optimal only if the social value coming from $\nu^{coop}$ is positive, i.e., if $\nu^{coop}(t_G)a_GV_G+\nu^{coop}(t_B)a_BV_B\geq 0$. In Example~\ref{example}, we derived $\nu^{coop}$ to be $\nu^{coop}(t_G)=1/3$ and $\nu^{coop}(t_B)=2/3$. Since $a_G=1/2$ and $a_B=3/4$, we are in case (i) of the preceding proposition if $\frac{1}{3}\cdot \frac{1}{2}\cdot V_G+\frac{2}{3}\cdot \frac{3}{4}\cdot V_B> 0$, i.e., $V_G> -3V_B$. (And if $V_G< -3V_B$, then any optimal team structure precludes all cooperation; if $V_G=-3V_B$, then optimal team structures have weight on $\nu^{coop}$ or on assignments that preclude cooperation.)

Suppose $V_G > - 3V_B$ and so Proposition~\ref{prop:task_assignment} implies that optimal team structure $\lambda$ puts maximal weight on $\nu^{coop}$ subject to satisfying the exogenous constraint on task coverage \eqref{eq:balance}. Recall that half the teams need to be assigned to the good task and half to the bad task, i.e., $q_G=q_B=1/2$. Hence, to maintain task coverage, the designer assigns some teams to only good tasks, $\nu^G$; Panel~(a) of Figure~\ref{fig:required_bad} depicts the relevant assignments. The relative weights of $\nu^{coop}$ and $\nu^G$ are found by noting that to cover all tasks, $\lambda$ must satisfy 
$
\lambda(\nu^{coop})\nu^{coop}(t_B)=q_B
$, giving $\lambda(\nu^{coop})=3/4$ (which implies $\lambda(\nu^{G})=1/4$; and there is no cooperation for any $r_{G}$).

For other task coverage constraints, the optimal team structure changes. For instance, suppose that at any given point in time, only 25\% of teams need to be assigned to the good task and now 75\% to the bad task, i.e., $q_G=1/4$ and $q_B=3/4$. Again, the designer assigns some teams to $\nu^{coop}$, reducing the share of bad tasks relative to the generalist assignment, and others now to $\nu^B$ to maintain task coverage; Panel~(b) of Figure~\ref{fig:required_bad} illustrates. The optimal team structure then satisfies 
$
\lambda(\nu^{coop})\nu^{G}(t_G)=q_G
$, giving $\lambda(\nu^{coop})=3/4$ (which implies $\lambda(\nu^{B})=1/4$; $r_{\nu^B}$ is sufficiently high to preclude cooperation).\footnote{To preclude cooperation on $\nu^B$, we require 
$
d_B > \frac{(1-r_{\nu^B})\delta}{1-(1-r_{\nu^B})\delta}a_Bc_B$,
i.e.,
$
\frac{1}{2} > \frac{(1-r_{\nu^B})\frac{3}{4}}{1-(1-r_{\nu^B})\frac{3}{4}}\cdot \frac{3}{4} \cdot 1,
$
which simplifies to 
$
r_{\nu^B} > \frac{7}{15}$.}

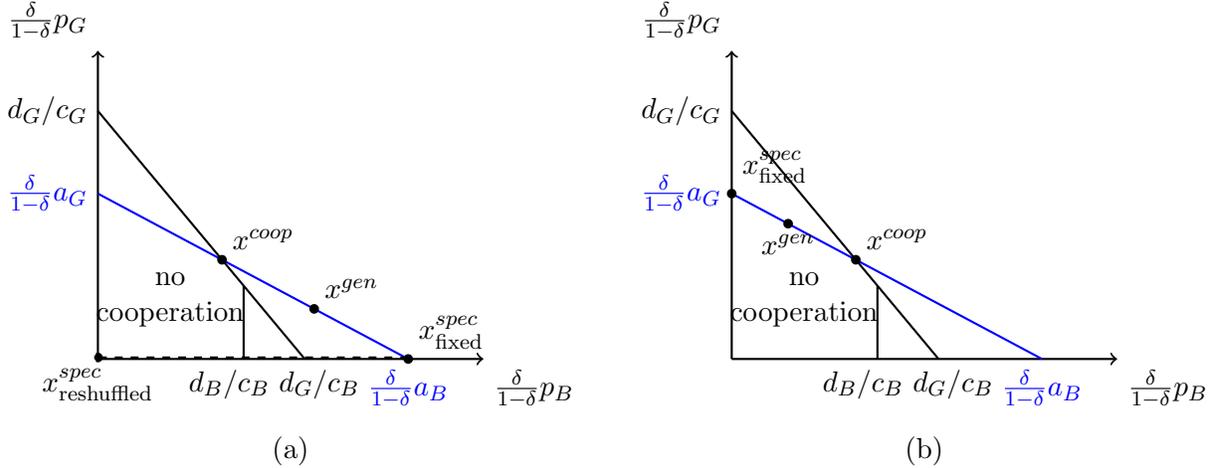
\begin{figure}[ht]
\small
\captionsetup{width=0.9\textwidth}
    \centering
    \begin{subfigure}[t]{0.5\textwidth}
        \centering
        \begin{tikzpicture}
            \def \xscale {.25}
            \def \yscale {.2}
            \draw[thick,->] (0*\xscale,0*\yscale) -- (20.5*\xscale,0*\yscale) node[anchor=north west] {$\frac{\delta}{1-\delta}p_B$};
            \draw[thick,->] (0*\xscale,0*\yscale) -- (0*\xscale,20.5*\yscale) node[anchor=south east] {$\frac{\delta}{1-\delta}p_G$};
            \draw[thick] (0*\xscale,16.5*\yscale) node[anchor=east] {$d_G/c_G$}--(11*\xscale,0*\yscale)  node[anchor=north,xshift=.5em] {$d_G/c_B$};
            \draw[thick,color=blue] (0*\xscale,11*\yscale) node[anchor=east] {$\frac{\delta}{1-\delta}a_G$} --(6.6*\xscale,6.6*\yscale);
           \draw[thick,color=blue] (6.6*\xscale,6.6*\yscale)--(16.5*\xscale,0*\yscale) node[anchor=north] {$\frac{\delta}{1-\delta}a_B$};
            \filldraw[thick] (3*\xscale,11*\yscale-3*11/16.5*\yscale) circle (.05cm) node[anchor=north] {$\nu^{gen}$};
            \filldraw[thick] (6.6*\xscale,6.6*\yscale) circle (.05cm) node[anchor=south west,align=center] {$\nu^{coop}$};
            \filldraw[thick] (6.6*\xscale,6.6*\yscale) node[anchor=south,align=center,xshift=3.5em,yshift=1 em] {};
            \filldraw[thick] (0*\xscale,11*\yscale) circle (.05cm) node[anchor=south west] {$\nu^{G}$};
            \draw (3.85*\xscale,1.5*\yscale) node[anchor=south,align=center] {no \\ cooperation};
            \draw[thick] (7.75*\xscale,16.5*\yscale-1.5*7.75*\yscale) -- (7.75*\xscale,0*\yscale) node [anchor=north,xshift=-.5 em] {$d_B/c_B$};
        \end{tikzpicture}
    \caption{\label{fig:unproductive_all_good} }
    \end{subfigure}%
    ~ 
    \begin{subfigure}[t]{0.5\textwidth}
        \centering
        \begin{tikzpicture}
            \def \xscale {.25}
            \def \yscale {.2}
            \draw[thick,->] (0*\xscale,0*\yscale) -- (20.5*\xscale,0*\yscale) node[anchor=north west] {$\frac{\delta}{1-\delta}p_B$};
            \draw[thick,->] (0*\xscale,0*\yscale) -- (0*\xscale,20.5*\yscale) node[anchor=south east] {$\frac{\delta}{1-\delta}p_G$};
            \draw[thick] (0*\xscale,16.5*\yscale) node[anchor=east] {$d_G/c_G$}--(11*\xscale,0)  node[anchor=north,xshift=.5 em] {$d_G/c_B$};
            \draw[thick,color=blue] (0*\xscale,11*\yscale) node[anchor=east] {$\frac{\delta}{1-\delta}a_G$} --(6.6*\xscale,6.6*\yscale);
           \draw[thick,color=blue] (6.6*\xscale,6.6*\yscale)--(16.5*\xscale,0*\yscale) node[anchor=north] {$\frac{\delta}{1-\delta}a_B$}; 
            \filldraw[thick] (11.5*\xscale,11*\yscale-2/3*11.5*\yscale) circle (.05cm) node[anchor=south west] {$\nu^{gen}$};
            \filldraw[thick] (6.6*\xscale,6.6*\yscale) circle (.05cm) node[anchor=south west,align=center] {$\nu^{coop}$};
            \filldraw[thick] (6.6*\xscale,6.6*\yscale) node[anchor=south,align=center,xshift=3.5 em,yshift=1 em] {};
            \filldraw[thick] (16.5*\xscale,0*\yscale) circle (.05cm) node[anchor=south west] {$\nu^{B}$};
            \draw (3.85*\xscale,1.5*\yscale) node[anchor=south,align=center] {no \\ cooperation};
            \draw[thick] (7.75*\xscale,16.5*\yscale-1.5*7.75*\yscale) -- (7.75*\xscale,0*\yscale) node [anchor=north,xshift=-.5 em] {$d_B/c_B$};
            \filldraw[thick, dashed] (16.5*\xscale,.1*\yscale) -- (0*\xscale,.1*\yscale) circle (.05cm) node [anchor=north] {$\nu^{B}_{\text{reshuffled}}$};
        \end{tikzpicture}
    \caption{\label{fig:unproducitve_all_bad} }
    \end{subfigure}%
        \caption{Optimal team structure candidates when bad games are necessary to sustain cooperation.}
 \label{fig:required_bad}\end{figure}
\end{continuance}

This discussion suggests that organizations may benefit by specializing their teams and reshuffling them at different rates. In the case of police officers tasked with either responding to emergencies or traffic policing, good cooperation among emergency responders---backing up one's partner in dangerous situations---may be supported by keeping officers together, and bad cooperation---colluding in taking bribes---can be stopped by reshuffling partnerships often.

\section{Dynamic, Reactive Assignments}
\label{sec:reactive}

In Section~\ref{sec:task_assignment}, teams could be (randomly) assigned to different tasks, but the assignment probabilities were the same in each period. We now further enrich the analysis to allow the designer to use assignments that change over time and possibly react to the past arrival process of games. 

As we show, a designer can improve incentives and subsequently the social value by using a time-varying assignment strategy. The basic idea is as follows. Suppose that to provide incentives to cooperate in good games, a team needs to expect to see a certain amount of bad games in the future. That is, they get more motivation from the anticipation of bad game arrivals than from good ones.  Under a random mixture, the possibility of seeing good games dampens the incentives from bad games. If, instead, once a team faces a good game, the designer promises to switch the team immediately to a task with only bad games for some time, then incentives to cooperate improve compared to having the team continue facing the mixture.

To illustrate this idea in more detail, consider again the parameterization in Example~\ref{example}.

\begin{continuance}{example}
Recall that $a_G=1/2,a_B=3/4$ and $c_G=c_B=1$. Hence, $a_Bc_B>a_Gc_G$, and so a future bad task provides stronger incentives to cooperate today than a future good task. (This is a general implication of our maintained assumption that $d_G>\frac{\delta}{1-\delta}a_Gc_G$ and $\max_g \{d_g\}<\frac{\delta}{1-\delta}a_Bc_B$.) 

The designer can use this fact to maximize incentives when they are most needed. Specifically, instead of using a fixed random assignment of $\nu^{coop}$ with $\nu^{coop}(t_G)=1/3$ and $\nu^{coop}(t_B)=2/3$, suppose the designer uses a carefully orchestrated sequential assignment,  so that the good task is always followed by two bad tasks and this then repeats---i.e., $t_G,t_B,t_B,t_G,t_B,t_B,t_G,t_B,t_B,\dots$. Crucially, this provides stronger incentives to cooperate on good games, while the same fraction of bad tasks is covered. If, in addition, $d_G>d_B$, and so given $\nu^{coop}$, players are indifferent between cooperating or not when facing a good game, and cooperation in the bad games is still strictly incentivized (i.e., $d_B$ is small enough), then players are now strictly incentivized to cooperate both after good and bad games. 

In the example presented here, players are, in fact, strictly incentivized to cooperate after a bad game since $d_B=\frac{1}{2}<\frac{\frac{3}{5}}{1-\frac{3}{5}}\frac{1}{2}\cdot 1=\frac{\delta}{1-\delta}a_Gc_G$. That is, even being assigned to only good tasks (or any mixed assignment of tasks) in the future would suffice to ensure cooperation in the current bad game. 

The designer can now use this slack in incentives to reduce the average number of bad tasks in this assignment, and, with it, the overall share of bad games with cooperation. Instead of following a good task with two bad tasks,  the designer follows a good task with one bad task for sure, and then a probability of a bad task below one (with the remaining probability returning to the good task early).  This uses fewer bad tasks, and the remaining bad tasks can be assigned to some constantly reshuffled team to avoid cooperation in the associated bad games and to maintain task coverage.  Overall, this is a strict improvement for the designer.
\end{continuance}

An application of a reactive assignment like that described in the example above is immediately assigning a military unit to noncombat status after it has faced combat and, after some  time, reassigning those soldiers back to the front where combat becomes more likely again. 

The analysis here is a first example of dynamically adaptive mechanism design, in which a designer changes the mechanism agents play over time in response to past events. We now formalize the benefits of this additional assignment flexibility and characterize when the designer gains from dynamically and reactively assigning teams to tasks, as in the running example. We then show that in a designer-optimal reactive team structure, teams can be grouped into two types---those that always (on-path) cooperate and those that do not---and characterize each type's assignments.

To formalize the ideas, we need to be explicit about what the designer can observe and condition their assignments on. To start, consider a situation in which the designer can observe whether a good game has arrived (but not which actions were played). For example, the designer can observe whether a team in the army engaged with the enemy, but not how they cooperated internally. It may be natural to also assume that the designer does not observe whether a bad game---e.g., an opportunity to take a bribe---arrived.  Observability of bad game arrivals turns out to be inconsequential in our results. After first treating the case where the arrival of a good game is observed, we also consider the case where the designer cannot observe whether a good game actually arrived in a given period.

\subsection{Definitions and a Preliminary Result}

A \emph{reactive assignment} is a function $m:(\mathcal{T}\times \{G,B,\emptyset\})^K\rightarrow \Delta(\mathcal{T})$, where $K$ is a positive integer. That is, instead of simply being a point in $\Delta(\mathcal{T})$, the function $m$ can condition on which previous tasks were assigned and which games, if any, were played up through the previous $K$ periods, and then (possibly randomly) assign a task. Let $\mathcal{M}$ denote the set of possible reactive assignments.

A reactive assignment results in a distribution over tasks that, when aggregated across a continuum of teams, appropriately staggered, leads to a steady-state frequency of tasks. Specifically, the function $m$ can also be viewed as a map from $(\mathcal{T}\times \{G,B,\emptyset\})^K$ into $\Delta\left((\mathcal{T}\times \{G,B,\emptyset\})^K\right)$, as it drops one entry and adds a new one each period.  For $s,s' \in \mathcal{S} \equiv (\mathcal{T}\times \{G,B,\emptyset\})^K$, let $m(s|s')$ be the probability of transitioning from \emph{state} $s'$ to $s$. Then $m$ defines a Markov chain on $\mathcal{S}$. A steady-state distribution induced by $m$ is denoted by $\overline{m}$ and satisfies $\overline{m}(s) = \sum_{s' \in \mathcal{S}}  m(s|s')\overline{m}(s')$. Note that since the Markov chain may be periodic or reducible---and thus admit multiple steady states---it is important to specify the relevant steady state.

To illustrate these definitions, consider again the assignment in Example~\ref{example} in which the good task is followed by one bad task for sure and then a bad task with probability $x$ (and with the remaining probability, $1-x$, a return to the good task). In steady state, the process induced by $m$ visits only a small number of aggregate states (equivalence classes of histories that are relevant for continuation play under this assignment). In this example, we can summarize the relevant aggregate states for the designer to track by (i) the current task, (ii) whether a game arrives in the current period, and (iii) whether the team is in the first or second bad-task position of the cycle. This yields six aggregate states: $(t_G,G)$, $(t_G,\emptyset)$, $(t_B,B)_1$, $(t_B,\emptyset)_1$, $(t_B,B)_2$, and $(t_B,\emptyset)_2$. We thus consider the induced Markov chain on these aggregates, since all histories within a given aggregate generate the same continuation distribution under this assignment. The associated transition matrix is
{\small \begin{equation}\label{eq:matrix} 
\bordermatrix{
   & (t_G,G) & (t_G,\emptyset) & (t_B,B)_1 & (t_B,\emptyset)_1 
     & (t_B,B)_2 & (t_B,\emptyset)_2 \cr
(t_G,G)        & 0 & 0 & a_B & 1-a_B & 0 & 0 \cr
(t_G,\emptyset)& 0 & 0 & a_B & 1-a_B & 0 & 0 \cr
(t_B,B)_1      & (1-x)a_G & (1-x)(1-a_G) & 0 & 0 & xa_B & x(1-a_B) \cr
(t_B,\emptyset)_1 
               & (1-x)a_G & (1-x)(1-a_G) & 0 & 0 & xa_B & x(1-a_B) \cr
(t_B,B)_2      & a_G & 1- a_G & 0 & 0 & 0 & 0 \cr
(t_B,\emptyset)_2 
               & a_G & 1- a_G & 0 & 0 & 0 & 0 \cr
}.
\end{equation}

For example, following a good task, teams move deterministically to the first bad task. Conditional on this transition, a bad game arrives with probability $a_B$, in which case the state becomes $(t_B,B)_1$; if no game arrives, the state is $(t_B,\emptyset)_1$. Then, teams transition to another bad task with probability $x$, and conditional on this transition, a bad game arrives with probability $a_B$, in which case the state becomes $(t_B,B)_2$, and with complementary probability, the state becomes $(t_B,\emptyset)_2$. With probability $1-x$, teams transition back to the good task, and similarly either to state $(t_G,G)$ or $(t_G,\emptyset)$, etc. The resulting steady-state distribution over the six states is 
\begin{equation}\label{eq:ss}
\overline{m}
= \frac{1}{2+x}\Big(
a_G,\; 1-a_G,\; a_B,\; 1-a_B,\; x a_B,\; x(1-a_B)
\Big).
\end{equation}

An organization's \emph{reactive team structure} is $\lambda \in \Delta(\mathcal{M})$, together with reshuffling rate $r_{m}$ and a steady-state distribution $\overline{m}$ for each $m \in \supp(\lambda)$. 
The history of a team and players' strategies, specified for each $m$, is as in the previous section. Furthermore, we continue to model reshuffling as a reduction in the discount factor, as, given the continuum of agents and a finite set of assignments and states, players can be matched to a new team in the same reactive assignment and state. 

Lastly, the dynamic team structure must respect a task-assignment constraint---the equivalent of \eqref{eq:balance}: for each task, the steady-state mass of teams assigned to states corresponding to that task, summed across reactive assignments, must satisfy an exogenous constraint. Formally, let $\mathcal{S}^t(m,\overline{m})\subseteq \mathcal{S}$ denote the set of states $s=[(t_{\tau},g_{\tau})]_{\tau=0}^K\in \mathcal{S}$ such that $t_K=t$ for $t\in \mathcal{T}$, i.e., those where the current task is task $t$. The overall constraint on a team structure is:
\begin{equation}\label{eq:reactiveconstraint}
    \sum_{m\in\supp(\lambda)} \lambda(m) \sum_{s \in \mathcal{S}^t(m,\overline{m})}  \overline{m}(s) = q_t \quad \text{for all $t\in \mathcal{T}$},
\end{equation}
for some $q_t$ with $\sum_{t\in \mathcal{T}}q_t=1$. 

In the assignment from Example~\ref{example}, using the steady-state distribution over states given in \eqref{eq:ss}, the share on task $t_G$ is $\frac{a_G}{2+x}+\frac{1-a_G}{x+2}=\frac{1}{x+2}$ (and consequently the share on task $t_B$ is $\frac{1+x}{2+x}$). Depending on the overall distribution of teams across tasks that must be satisfied and the exact value of $x$, a team structure $\lambda$ may have to put weight on other reactive assignments in order to satisfy \eqref{eq:reactiveconstraint}.

This extension of the model, which allows the designer a richer set of assignments, again nests all previous formulations. For instance, a random, but static, task assignment as in Section~\ref{sec:task_assignment} is captured by reactive assignment $m$ being a constant function.

The next ``technical'' lemma establishes that Properties~\ref{property:symmetry}--\ref{property:star} continue to imply a simple on-path equilibrium structure in which players only fully cooperate or fully not cooperate. This simplifies the analysis, as we do not need to consider more intricate potential structures in which some players cooperate, and others do not.  Importantly, this result implies that there is a well-defined social value given any reactive team structure.

Given a reactive assignment $m$ with reshuffling rate $r_m$ and steady-state distribution $\overline{m}$, we say that an equilibrium strategy profile has \emph{cooperation in state $s=[(t_{\tau},g_{\tau})]_{\tau=0}^K\in \mathcal{S}$} if, on the equilibrium path, all players cooperate in $g_K$, the current game; analogously, we say there is \emph{no cooperation in state $s$} if all players do not cooperate. 

\begin{lemma}\label{claim:pure}
    Consider a reactive assignment $m$ with reshuffling rate $r_m$ and steady-state distribution $\overline{m}$. 
    In any player-optimal subgame-perfect equilibrium and state $s\in \mathcal{S}$, either there is cooperation in state $s$, no cooperation in state $s$, or no game occurs in $s$ (i.e., $g_K=\emptyset$).
\end{lemma}

Lemma~\ref{claim:pure} implies a well-defined social value of a reactive team structure. Let $\mathcal{G}(m,r_m,\overline{m})\subseteq \mathcal{S}$ (respectively, $\mathcal{B}(m,r_m,\overline{m})$) be the set of states $s$ such that there is cooperation in $s$ and the current game in $s$ is the good game (respectively, bad game). The social value is then given by 
\begin{equation}\label{eq:obj}
    \sum_{m \in \supp(\lambda)}\lambda(m)\left(\sum_{s\in \mathcal{G}(m,r_m)}\overline{m}(s)V_G+\sum_{s\in \mathcal{B}(m,r_m)}\overline{m}(s)V_B\right).
\end{equation}

To illustrate, the social value coming from our running example of a reactive assignment with steady state given in \eqref{eq:ss} is proportional to $\frac{a_G}{2+x}V_G+\left(\frac{a_B}{2+x}+\frac{xa_B}{2+x}\right)V_B$, assuming there is cooperation throughout.

Next, we consider which reactive team structures maximize~\eqref{eq:obj} subject to the constraint on the marginal task distribution, \eqref{eq:reactiveconstraint}.

\subsection{Designer-Optimal Assignments}

Throughout what follows, we focus on the most relevant case in which the incentive to deviate in bad games, $d_B$, is sufficiently low so that the incentives provided to cooperate in good tasks are the binding constraint.\footnote{When, instead, $d_G$ is sufficiently low, full specialization ensures optimal cooperation (Observation~\ref{obvs:full_specialization}).} 

In this section, we first ask when the designer benefits from the added flexibility of dynamically and reactively assigning teams to tasks. We then show that optimal reactive team structures take a particular form. We characterize these structures both when the designer can condition task assignments on the occurrence of games and when they cannot. In an optimal team structure, teams that cooperate are assigned to the good task until a good game arrives, at which point they are moved to the bad task for some amount of time before returning to the good task. Teams that do not cooperate pick up the slack to satisfy the exogenous task constraint.  Finally, we examine how our results change as the length of a time period shrinks to zero. 

When does dynamically and reactively assigning teams to tasks increase the optimal social value? If full specialization leads to optimal cooperation (Observation~\ref{obvs:full_specialization}), then there is no room to improve the social value further. At the other extreme, if there is no share of bad tasks that ensures cooperation, then cooperation on good games is never possible, even with a dynamic assignment. Thus, we focus on situations in which partial specialization leads to optimal outcomes and cooperation on good games can be incentivized, as was illustrated in Figure~\ref{fig:required_bad}. In our running example, a further improvement to the static assignment $\nu^{coop}$ was possible as $d_G<d_B$.  The next proposition shows that the general condition is unequal deviation payoffs, $d_G\neq d_B$. 

\begin{proposition}\label{prop:when}
    Suppose that $d_G> \frac{\delta}{1-\delta}a_G c_G$ and $\max_g \{d_g\} < \frac{\delta}{1-\delta} a_B c_B$, and that there exists an optimal (static) team structure $\lambda$ that has positive mass only on $\nu^{coop}$ and $\nu^g$ for some $g\in\{G,B\}$ where:
    \begin{itemize}
        \item $\nu^{coop}$ is as defined in Equation \ref{hybrid} and $r_{\nu^{coop}}=0$
        \item and $\nu^{g}(t_{g})=1$ and $r_{g}$ is sufficiently high to preclude cooperation if $g=B$.
    \end{itemize} Then, the social value of the optimal reactive team structures is strictly higher than that of optimal nonreactive team structures if $d_G\neq d_B$.
\end{proposition}

Given the premise of the proposition and $d_G\neq d_B$, players have a strict incentive to cooperate on one of the games, say game $g$, and are indifferent between cooperating or not on the other, say game $g'$. The designer can then make future bad games, those that provide relatively stronger incentives for current cooperation, more likely following game $g'$ and less likely following game $g$, keeping the unconditional distribution of games (and tasks) unchanged. As a result, players are now strictly incentivized to cooperate on both games, allowing the designer to reduce the share of bad tasks in that assignment, as in the running example, and thereby improve the social value.\footnote{The converse of the proposition statement also holds, but only for the case when the designer does not observe the arrival of games.}

We now turn to characterizing designer-optimal reactive team structures. 

Our first result makes it easier to describe the optimal structure as it establishes that there are essentially two types of teams: those that always cooperate and those that never cooperate. 

\begin{lemma}\label{lem:design_opt}
Suppose that $d_G> \frac{\delta}{1-\delta}a_G c_G$ and $\max_g \{d_g\} < \frac{\delta}{1-\delta} a_B c_B$. Consider a designer-optimal reactive team structure $\lambda$, $m\in \supp (\lambda)$ with a reshuffling rate $r_m$ and steady state $\overline{m}$. For any player-optimal equilibrium, all players in $m$ either cooperate at all on-path histories $h$ and games $g$ (i.e., they are \emph{fully cooperative}), in which case $r_m=0$, or all players do not cooperate at all on-path histories $h$ and games~$g$ (i.e., they are \emph{fully not cooperative}). 
\end{lemma}

For intuition, suppose a type of assignment in the support of $\lambda$ induces cooperation only sometimes. The designer could split such  an assignment into two assignments: one that transitions only across the states with cooperation in the initial assignment and another that is assigned to the remaining states. Since there would be no periods without cooperation, teams with the first assignment are now strictly incentivized to cooperate, allowing the designer to reduce the share of bad games with cooperation, and Lemma~\ref{lem:design_opt} follows. Thus, organizations optimally minimize the share of undesirable cooperation needed to sustain desirable cooperation by keeping cooperating teams together and eliminating states without cooperation, thereby allowing the incentives arising from future undesirable cooperation to have their maximum impact. 

\subsubsection{Good Games Are Observable}\label{sec:obs}

Now, we fully characterize the optimal reactive team structures, beginning with the case in which the designer observes whether a good game occurs.

To understand why observability of the good game's arrival matters, let us revisit the dynamic assignment in which a good task is followed by a bad task, followed by a bad task with some probability before repeating. The designer can further improve by conditioning on whether a good game actually arrives under the good task. If none arrive, then there is no reason to transition to the bad task, as no incentives are needed. In contrast, if one arrives, then the bad task assignment is needed. Thus,  the designer can improve by waiting to transition to the bad task until after a good game actually arrives.   

For this analysis, we specialize to the interesting case in which the incentive for cooperation in the good game is the binding constraint. In particular, we assume that the temptation to deviate on a current bad game is sufficiently low, specifically that facing a stream of future good games alone would suffice to ensure cooperation on current bad games, i.e., $d_B<\frac{\delta}{1-\delta}a_Gc_G$. When this condition holds, a team would cooperate in equilibrium in the bad game, regardless of which games it sees in the future (presuming cooperation). Given this assumption, the binding incentive constraint is on good games, and the ensuing proposition provides a full characterization of designer-optimal team structures. 

By Lemma~\ref{lem:design_opt}, there are only \emph{fully cooperative} or \emph{fully non-cooperative} teams. Proposition~\ref{prop:main} below additionally says that if a fully cooperative team exists, it is assigned to the good task until a good game arrives; then, it is assigned to the bad task for a number of periods before returning to the good task.\footnote{While it may take an arbitrary number of periods to transition from the good task to the bad task, a reactive assignment conditioning on a finite number of past tasks and games suffices, as only whether a good game arrived in the last period matters.} The duration of the assignment to the bad tasks is minimized, subject to providing enough incentives to cooperate on the good game. If the resulting steady-state distribution of good and bad games with cooperation leads to a net positive social value, the mass of teams in that reactive assignment is maximized, and the teams that never cooperate are assigned to the remaining task to maintain the overall condition on task coverage, as before.

Because reactive assignments allow arbitrary (random) conditioning on whatever the designer observes, the policy class can be extremely rich. Our characterization pins down the global optimum in this high-dimensional space. It shows that optimality is attained by a simple policy with only two building blocks (fully cooperative vs.\ fully non-cooperative teams) and a minimal sequence of bad tasks after good arrivals.

\begin{proposition}\label{prop:main}
Suppose that $d_G> \frac{\delta}{1-\delta}a_G c_G$, $\max_g \{d_g\} < \frac{\delta}{1-\delta} a_B c_B$, $d_B<\frac{\delta}{1-\delta}a_Gc_G$, and that the designer observes when a good game occurs. Consider a designer-optimal reactive team structure $\lambda$. For each $m \in \supp(\lambda)$ in which players are fully cooperative, the following holds:
    \begin{enumerate}
        \item When playing a good game, players are indifferent between cooperating and deviating and not cooperating.
        \item After playing a good game, teams are assigned the bad task for $N_B$ periods; in the period after that, teams must be assigned the good task with positive probability and then probability 1 in the next period if they are not assigned to it in that period.
        \item If players are assigned to the good task, they continue to be assigned to the good task until they play the good game.
    \end{enumerate}

    Furthermore, unless there is an optimal reactive team structure $\lambda$ with fully non-cooperative teams for all reactive assignments $m\in\supp(\lambda)$, all reactive assignments $m'\in \supp(\lambda)$ in which players are fully non-cooperative assign teams to the same task.
\end{proposition}

To illustrate the proposition, consider again Example~\ref{example}. The proposition implies that the reactive assignment in which players always cooperate, and that forms part of a designer-optimal reactive team structure, can be described by the following Markov chain
{\small $$
\bordermatrix{
   & (t_G,G) & (t_G,\emptyset) & (t_B,B)_1 & (t_B,\emptyset)_1 
     & (t_B,B)_2 & (t_B,\emptyset)_2 \cr
(t_G,G)        & 0 & 0 & a_B & 1-a_B & 0 & 0 \cr
(t_G,\emptyset)& a_G & 1-a_G & 0 & 0 & 0 & 0 \cr
(t_B,B)_1      & (1-x)a_G & (1-x)(1-a_G) & 0 & 0 & xa_B & x(1-a_B) \cr
(t_B,\emptyset)_1 
               & (1-x)a_G & (1-x)(1-a_G) & 0 & 0 & xa_B & x(1-a_B) \cr
(t_B,B)_2      & a_G & 1- a_G & 0 & 0 & 0 & 0 \cr
(t_B,\emptyset)_2 
               & a_G & 1- a_G & 0 & 0 & 0 & 0 \cr
}.
$$}

Specifically, consider a team currently assigned to the good task. If no game arrives, the state is $(t_G,\emptyset)$, and it is assigned again to the good task. Conditional on this transition, a good game arrives with probability $a_G$, in which case the state becomes $(t_G,G)$; otherwise, the state remains $(t_G,\emptyset)$. If the state is $(t_G,G)$, the team is then assigned to the bad task for at least one period. In particular, if a bad game arrives, the next state is $(t_B,B)_1$, and if no bad game arrives, the next state is $(t_B,\emptyset)_1$. As in the transition matrix \eqref{eq:matrix}, the team may then be assigned to another bad task with probability $x$, or transition back to the good task. For the parameter values in Example~\ref{example}, one can show that, indeed, for $N_B=1$ and $x=5/9$, players are indifferent between cooperating and deviating when facing a good game, as optimality requires.

Why is it not optimal to have a reactive assignment (in which players always cooperate) in which a good task with no arriving good game is followed by a bad task? Recall that bad tasks provide greater incentives to cooperate than good tasks, that is $a_Bc_B>a_Gc_G$. While bad tasks are necessary to provide sufficient incentives for cooperation, the designer uses them as sparingly as possible. Thus, an optimal assignment assigns teams to bad tasks when the temptation to deviate is greatest, and this is following a good game. To optimally leverage bad tasks to provide the most incentives, the bad tasks are ``frontloaded.'' Specifically, the delay with which bad games arrive after a good game and before the next good game must be as short as possible for optimality.

We emphasize that the optimal reactive team structure does not need to condition on whether a bad game occurred. This implies that the designer observing the occurrence of a bad game does not improve the social value attainable.  This follows because there is no need to incentivize cooperation in bad games,  and it is only bad games' \emph{expected} arrivals that matter in incentivizing cooperation in good games. However, the optimal team structure does condition on the occurrence of a good game---teams remain assigned to the good task until a good game arrives. 

\subsubsection{Good Games Are Unobservable}\label{sec:noobs}

In this section, we fully characterize the optimal reactive team structures for the case in which the designer does not observe whether a good game occurs when a team is assigned to the good task. 

Again, revisiting Example~\ref{example}, the improvement we noted at the beginning of Section~\ref{sec:obs} in which the designer waited for the arrival of a good game before transitioning is no longer possible. Instead, the designer simply transitions teams directly after the good task, as the incentives are needed if a good game did indeed arrive; that is, the teams transition according to the matrix \eqref{eq:matrix}.  As a result, this case does strictly worse than the case in which the arrival of the game is observable (and arrives with probability less than one), but still strictly improves over the case with a static assignment. 

Proposition~\ref{prop:main_g_unob} is the analog of Proposition~\ref{prop:main} for the case in which the designer does not observe good games.

\begin{proposition}\label{prop:main_g_unob}
Suppose that $d_G> \frac{\delta}{1-\delta}a_G c_G$, $\max_g \{d_g\} < \frac{\delta}{1-\delta} a_B c_B$, $d_B<\frac{\delta}{1-\delta}a_Gc_G$, and that the designer does not observe when a good game occurs. Consider a designer-optimal reactive team structure $\lambda$. For each $m \in \supp(\lambda)$ in which players are fully cooperative, the following holds:
\begin{enumerate}
    \item When playing a good game, players are indifferent between cooperating and deviating and not cooperating.
    \item After being assigned the good task, teams are assigned the bad task for $N_B$ periods; in the period after that, teams must be assigned the good task with positive probability (and then probability 1 in the next period if they are not assigned to it in that period.
\end{enumerate}

Furthermore, unless there is an optimal reactive team structure $\lambda$ with fully non-cooperative teams for all reactive assignments $m\in\supp(\lambda)$, all reactive assignments $m'\in \supp(\lambda)$ in which players are fully non-cooperative assign teams to the same task.
\end{proposition}

The proposition confirms that the Markov Chain transition matrix \eqref{eq:matrix} characterizes the structure of an optimal reactive assignment when cooperation is optimal and good games are unobservable. Given the parameters in Example~\ref{example}, we obtain $N_B=1$ and $x=5/27$. In the corresponding steady state, fully cooperative teams spend $32/59$ of their time on bad tasks, compared to $7/16$ when the designer observes the arrival of good games. Thus, social value is strictly lower in the unobservable case (which still is an improvement relative to the static case in which the share spent on bad tasks is $\nu^{coop}(t_B)=2/3$).

Note that, as before, the reactive assignments do not condition on whether a bad game occurred. The designer observing the occurrence of a bad game still does not improve the social value. 

In both the observable and unobservable cases, optimal reactive team structures have a simple form: they rely on only two building blocks---fully cooperative and fully non-cooperative teams---and feature a finite sequence of bad-task assignments following good-task assignments. The key difference is that when good games are observable, the transition to the bad task can be conditioned on a good game's actual arrival, so it is not forced when it is not needed, whereas when the arrival is unobservable, the designer must transition regardless. Consequently, the observable case achieves the same incentive provision with fewer bad tasks in steady state and strictly higher social value.

\section{Concluding Remarks}

A systemic approach is important in understanding and designing organizations. Nominally independent tasks and situations become dependent if team members interact repeatedly. Our analysis suggests that organizations can overcome some of the undesirable spillover effects through intentional design.  For example, the \cite{mollen1994}, discussing corruption and policing, describes an empirical regularity showing the importance of a systemic view: ``[the code of silence] is strongest where corruption is most frequent. This is because the loyalty ethic is particularly powerful in crime-ridden precincts where officers most depend upon each other for their safety each day.'' The ``loyalty ethic'' corresponds to cooperation in bad situations, and the ``depend[ence] upon each other for their safety'' describes the benefits from cooperation in good situations. The report then suggests that ``crime-ridden precincts'' are those with a high frequency of the good game (e.g., more dangerous situations), which facilitates cooperation in bad games---an important instantiation of spillover effects. Our results on reshuffling, specialization, and dynamic assignments suggest organizational design can be used to address these concerns.  The police force can specialize teams that encounter dangerous situations on the job (who then also accept bribes) and thereby reduce the number of officers in such situations to minimize corruption. The police force can optimize further by rotating such teams to tasks that foster teamwork, building up a norm of cooperation before heading back ``into the field.'' 

Our analysis also relates to the cohesion of communities and societies. \cite{banfield1958,putnam2000}, in their studies of an Italian city and the US over time, respectively, document these spillovers. \cite{banfield1958} describes communities with nuclear families as uncooperative actors. As Lemma~\ref{lem:equilibrium_characterization} shows, a lack of interaction in some dimensions may lead to a breakdown of all cooperation (public goods, commons, etc.). Similarly, the resulting need for cooperation within the family can create powerful incentives to act as a single unit, which could enhance the ability to coordinate on being corrupt. Relatedly, \cite{putnam2000} warns about the consequences of the erosion of civic interactions (``good cooperation'') in the US for other kinds of institutions that require cooperation, e.g., political discourse. Lemma~\ref{lem:equilibrium_characterization} supports his warning of spillovers from cooperation on leisurely activities, such as bowling, to more general foundations of civic behavior. It further provides some insight into what may have caused such a decline in civic norms: formal institutions, such as insurance markets and markets, replacing informal interactions---such as reciprocal insurance and favor exchange---thereby reducing necessary cooperation. 

\section{Bells and Whistles}\label{sec:bells}

In closing, we discuss a few extensions and other variants of the model. 

\paragraph{Hybrid tasks.} When assigning teams to tasks, we focused on tasks that either had only good games or only bad games. Hybrid tasks that involve some random game arrivals could be accommodated, albeit with complications in notation and cases, but would yield similar results throughout. Also, in our model, the set of tasks is fixed. However, when supporting optimal cooperation is challenging, it can be advantageous to introduce new tasks that provide additional benefits from cooperation. These new tasks might not benefit society or directly increase organizational productivity, but they might help sustain cooperation in good games. For instance, scheduling ``team-building'' exercises, social events, and non-production related contests in which team members must cooperate with each other can help induce cooperation in other games as well.   Many organizations use such team-building to enhance the loyalty that people feel towards the units in which they are involved, fostering cooperation.  Although such events may be time-consuming and thus reduce the time that teams spend directly on productive activities, they may end up enhancing productivity and substitute for bad games by enabling cooperation in good games when that might have otherwise been infeasible.

\paragraph{Heterogeneity.} There may also be individual heterogeneity of team members in their preference to engage in collusive behavior. Adding individuals with a strong inclination to defect on bad games could curb collusion, but could also undermine productive cooperation. Similarly, ``changing the culture'' of an organization and increasing awareness of the negative consequences of collusion---that is, shifting team members' preferences to be more averse to collusion on bad tasks---may erode the informal incentives that sustain cooperation on good tasks and productive teamwork. 

\paragraph{Learning and coordination.} We have focused on the incentives for cooperation in our analysis.  There can be other benefits from keeping teams together, such as players learning from each other and anticipating each other's behaviors, etc.  This provides additional reasons for keeping teams together, which could be incorporated into an extended model.  

\paragraph{Monitoring.} Lastly, surveillance and other technologies also play an important role and may be modeled explicitly. Improved monitoring can make cooperation directly enforceable in good games and punishable in bad ones, thereby making incentive issues less relevant. Thus, organizations can benefit from investing in technologies to monitor behavior. For example, if cooperation in the good game can be directly enforced by a monitoring technology, reshuffling teams to reduce residual collusion makes good sense. However, Section~\ref{sec:task_assignment} shows that an organization can benefit from the games' strategic interdependencies. For example, consider an organization that prefers total cooperation to no cooperation. If no cooperation in the bad game is directly enforced via some new technology, then cooperation in the good game may be impaired, as players can now only punish deviation from cooperation via that game and not the bad game. All cooperation may thereby break down.

\paragraph{The continuous time limit.} Time periods in our model are discrete, and a natural question to ask is how the previous results depend on the length of a time period. To this end, we consider the continuous limit of our model; i.e., when time periods become infinitesimally small and, with them, the probability of a game occurring and the potential frequency of reassignments. We will see that whether the designer can improve upon static assignments as the period length shrinks with a dynamic assignment depends on whether the arrival of good games is observable. This suggests that dynamic assignments in the unobservable case are only useful in situations where at most one good game arrives over a nontrivial span of time.

So far, a period has been of length $1$. Suppose now that a period is of length $T<1$, which implies two changes. First, players in period $\tau$ discount payoffs received in period $\tau'>\tau$ by less, specifically, by $\delta^{(\tau'-\tau)T}$. Second, the probabilities of arrival of good and bad games on the two tasks are also reduced, specifically to $a_GT$ and $a_BT$, respectively; and, as $T\rightarrow 0$, games occur according to a Poisson process. As before, the designer can change assignments dynamically and reactively in each period.\footnote{There might be institutional constraints that prevent the designer from adjusting assignments every period.  We do not study this variant, but we do not believe it will change the basic intuition of the results, as it just adds constraints to the designer.}

We first note that as $T$ decreases, while the designer's payoff from any fixed static team structure stays the same, the designer's payoff from an \emph{optimal} static team structure increases. The expected arrival rate of games remains unchanged, but the timing of arrivals is more continuous rather than concentrated at the end of a full period. Under exponential discounting, a continuous stream of possible arrivals within, say, the next two periods is valued more highly than a single arrival exactly one period away. Thus, the same expected rate of future cooperation is effectively worth more. As a result, the designer can sustain cooperation with a smaller share of bad tasks, and the social value under the optimal static assignment increases as $T \to 0$.

Now consider reactive team structures and the case where the designer observes good games occurring (Section~\ref{sec:obs}). The optimal reactive team structure does not change substantively: consider a reactive assignment in which players always cooperate, which is part of an optimal reactive team structure. In the limit as $T\rightarrow 0$, following a good game, teams spend some fixed amount of time on the bad task before returning to the good task, where they wait until a good game arrives, as before. Furthermore, as was true in the $T=1$ case, the social value is higher when reactive assignments are possible compared to static assignments, as the designer can wait for the arrival of a good game.\footnote{Thus, combining this conclusion with the observation in the previous paragraph, both the reactive and static assignments are improving, but the reactive assignment remains a strict improvement.}

Consider next an optimal reactive team structure when the designer does not observe good games occurring (Section~\ref{sec:noobs}) and a reactive assignment in which players always cooperate. As $T\rightarrow 0$, players are immediately transitioned out of the good task, and the time spent on ensuing bad tasks also becomes arbitrarily short. In other words, players are randomly assigned to the good or bad tasks with time-independent probabilities, and the designer cannot benefit from a reactive assignment. Intuitively, if time periods are long, then there is a high chance that a team on the good task faced a good game. Thus, it is as if the designer knows a good game took place. When the length of a period shrinks, the designer loses this knowledge. As a result, making bad games relatively more likely following a good game becomes infeasible. This implies that the gains the designer has from reactive assignments compared to static assignments when good games are not observed vanish as the length of the period goes to zero and arbitrarily frequent adjustments are possible.\footnote{One can show that the designer's payoff improves as $T\rightarrow 0$. The designer's payoff from a static assignment also improves as $T\rightarrow 0$,  and it is the relative improvement that vanishes.} Thus, whether there is a non-vanishing improvement over the static assignment as the time period shrinks depends on the observability of the arrival of good games.

{
\small 
\bibliographystyle{aer}
\bibliography{teamsorganizationculture}
}

\appendix

\section{Proofs}\label{appendix:proofs}

\begin{proof}[Proof of Lemma~\ref{lem:equilibrium_characterization}] 
Suppose \eqref{eq:full_cooperation} holds.  Consider the grim-trigger strategy: `Always cooperate unless any player did not cooperate in the past, at which point always do not cooperate.' Property~\ref{property:max} states that on-path aggregate stage-game payoffs are maximized when every player chooses this strategy. Furthermore, the resulting strategy profile constitutes an equilibrium: Facing game $g$, no player benefits from deviating since not cooperating leads to an increase in the stage-game payoff of $d_g$ and a decrease in future payoffs of $\sum_{g'} \frac{\delta}{1-\delta}p_{g'}c_{g'}$ and by \eqref{eq:full_cooperation} and since not cooperating if someone deviated in the past is a Nash equilibrium in any stage game (an implication of Property~\ref{property:dom}), i.e., the punishment stage is a Nash equilibrium.

Next, suppose \eqref{eq:partial_cooperation} holds for game $g$ but \eqref{eq:full_cooperation} fails. Call the other game $g'$. First, we show that an equilibrium exists with on-path cooperation only in game $g$. Consider the grim-trigger strategy: `Cooperate in $g$ only unless any player deviated from this strategy in the past, in which case do not cooperate in either game'. Facing game $g$, no player benefits from deviating since not cooperating leads to an increase in stage-game payoff of $d_g$ and a decrease in future payoffs of $\frac{\delta}{1-\delta}p_gc_g$ and by \eqref{eq:partial_cooperation}. There is also no benefit from deviating in game $g'$ as all players not cooperating is a Nash equilibrium, and future payoffs are lower after deviation. And lastly, as before, the punishment stage is a Nash equilibrium. We also need to show that there is no other equilibrium that yields higher aggregate payoffs. As before, cooperating in game $g$ generates the highest aggregate payoffs in that game (Property~\ref{property:max}); no cooperation in game $g'$ maximizes aggregate payoffs in that game, provided cooperation is not possible (Property~\ref{property:last}). Thus, we need to check whether there exists an equilibrium in which players may cooperate in game $g'$. To this end, consider any history and suppose players face game $g'$. By Property~\ref{property:star}, a player not cooperating leads to an increase in their stage-game payoff of at least $d_{g'}$ relative to cooperating. By Properties~\ref{property:max} and~\ref{property:minmax}, there exists a player $i$ whose future payoffs decrease by at most $\frac{\delta}{1-\delta}(p_gc_g+p_{g'}c_{g'})$ when not cooperating relative to cooperating. As  \eqref{eq:partial_cooperation} holds for game $g$ but \eqref{eq:full_cooperation} fails, we have $d_{g'}=\max \{d_g,d_g'\}$, and thus player $i$ strictly prefers to not cooperate. Hence, there can be no equilibrium in which players may cooperate in game $g'$.

Lastly, suppose that we are not in either of the above cases. First, note that `Always do not cooperate' constitutes an equilibrium.  We show that there is no other equilibrium that yields higher aggregate payoffs. As in the previous (second) case, a higher aggregate payoff requires an equilibrium in which players may cooperate in some game $g$ (Property~\ref{property:last}). If there may also be cooperation in the other game, call it $g'$, then arguments analogous to those for the first case imply that such cooperation requires $\max \{d_g,d_g'\}\leq \frac{\delta}{1-\delta}(p_gc_g+p_{g'}c_{g'})$; if there is no cooperation in game $g'$, then arguments analogous to those for the second case imply that $d_g\leq \frac{\delta}{1-\delta}p_gc_g$. Either way, \eqref{eq:partial_cooperation} not holding and  \eqref{eq:full_cooperation} not holding for both games imply there cannot be cooperation in equilibrium. 
\end{proof}

\begin{proof}[Proof of Proposition~\ref{prop:optimal_reshuffling}]
Suppose conditions~\ref{cond:deviation}--\ref{cond:relative} given in the propositions hold. By condition~\ref{cond:absolute}, there exists a reshuffling rate $r\in[0,1]$ that solves $d_G=\frac{\delta^e}{1-\delta^e}p_Gc_G$, where $\delta^e\coloneqq \delta(1-r)$. Thus, \eqref{eq:partial_cooperation} holds for the good game. What remains to show is that \eqref{eq:full_cooperation} does not hold. We have
\begin{equation}\label{eq:proof_optimal_reshuffling}
    \frac{\delta^e}{1-\delta^e}(p_Gc_G+p_Bc_B)=d_G+\frac{d_G}{p_Gc_G}p_Bc_B<d_B=\max\{d_G,d_B\},
\end{equation}
where the first equality follows from the definition of $r$ (and $\delta^e$), the inequality follows from conditions~\ref{cond:deviation} and~\ref{cond:relative}, and the second equality again from condition~\ref{cond:deviation}. Hence, \eqref{eq:full_cooperation} does not hold, and there is only cooperation in the good game by Lemma~\ref{lem:equilibrium_characterization}.

Suppose now not all of conditions~\ref{cond:deviation}--\ref{cond:relative} hold. We need to show that optimal cooperation is not possible. If condition~\ref{cond:deviation} fails, then \eqref{eq:partial_cooperation} holding for the good game implies \eqref{eq:full_cooperation} for any $r$; thus, there cannot be optimal cooperation. If condition~\ref{cond:absolute} does not hold, then \eqref{eq:partial_cooperation} does not hold either for the good game for any $r$, and again there cannot be optimal cooperation. Lastly, suppose condition~\ref{cond:relative} does not hold.  We need to show that \eqref{eq:partial_cooperation} for the good game cannot hold while \eqref{eq:full_cooperation} fails for any reshuffling rate $r\in [0,1]$. As the right-hand sides of \eqref{eq:full_cooperation} and \eqref{eq:partial_cooperation} decrease in $r$ while the respective left-hand sides are independent of it, it suffices to consider the largest $r$ such that \eqref{eq:partial_cooperation} holds, that is, the $r$ that solves \eqref{eq:partial_cooperation} for the good game as an equality. A set of equations analogous to those in~\eqref{eq:proof_optimal_reshuffling} with the inequality reversed (and becoming a weak inequality) as condition~\ref{cond:relative} fails then shows that \eqref{eq:full_cooperation} indeed does not hold.
\end{proof}

\begin{proof}[Proof of Observation~\ref{obvs:full_specialization}]
Consider the full-specialization team structure $\lambda$ with $\lambda(\nu^g(t_g))=q_g$ for $t\in\{G,B\}$, $r_G=0$, and $r_B=1$, and first suppose $d_G\leq \frac{\delta}{1-\delta}a_Gc_G$. Consider a team assigned to $\nu^G$. Property~\ref{property:max} implies that on-path aggregate stage-game payoffs are maximized when every player chooses grim-trigger strategy `Always cooperate unless any player did not cooperate in the past, at which point always do not cooperate.' Furthermore, the resulting strategy profile constitutes an equilibrium: Facing the good game, the only game players of such team face, no player benefits from deviating since not cooperating leads to an increase in the stage-game payoff of $d_G$ and a decrease in future payoffs of $\frac{\delta}{1-\delta}a_Gc_G$ and by assumption and since not cooperating if someone deviated in the past is a Nash equilibrium (an implication of Property~\ref{property:dom}), i.e., the punishment stage is a Nash equilibrium. Consider now a team assigned to $\nu^B$. As $r_B=1$, players' discount factor is $0$ and so they play the unique stage game equilibrium, which is no cooperation. Thus, there is cooperation only on good games.

Suppose now that $d_G>\frac{\delta}{1-\delta}a_G c_G$, and again consider a team assigned to $\nu^G$. Every player choosing strategy `Always not cooperating' constitutes an equilibrium as every player not cooperating is stage-game equilibrium. Using Properties~\ref{property:max} and~\ref{property:last}, to see whether this equilibrium maximizes maximizes on-path aggregate stage-game payoffs, we need to check whether there exists an equilibrium in which players may cooperate. To this end, consider any history and suppose players face the good game. By Property~\ref{property:star}, a player not cooperating leads to an increase in their stage-game payoff of at least $d_G$ relative to cooperating. By Properties~\ref{property:max} and~\ref{property:minmax}, there exists a player $i$ whose future payoffs decrease by at most $\frac{\delta}{1-\delta}a_Gc_G$ when not cooperating relative to cooperating. Since $d_G>\frac{\delta}{1-\delta}a_G c_G$ by assumption, player $i$ has a profitable deviation and there can be no equilibrium in which players may cooperate. 
\end{proof}

\begin{proof}[Proof of Proposition~\ref{prop:task_assignment}]
Suppose not every team structure that has only assignments that lead to no cooperation is optimal, i.e., we are not in case (ii). Consider an optimal team structure $\lambda$ with reshuffling rates $r_\nu$ for $\nu \in \supp(\lambda)$. We have to show that $\lambda$ has positive mass only on the two assignments specified in the proposition. 

Since we are not in case (ii) and since all team structures that only have assignments that lead to no cooperation yield the same social value, 0, no such team structure is optimal, and there exists an assignment $\nu \in \supp(\lambda)$ with $r_\nu$ such that there is cooperation.  We first show that $\nu=\nu^{coop}$ and $r_{\nu}=r_{\nu^{coop}}=0$, where $\nu^{coop}$ is defined in \eqref{hybrid}. First, if $r_{\nu^{coop}}>0$, there is no cooperation, so if $\nu^{coop}\in \supp(\lambda)$ and $r_{\nu^{coop}}>0$, then it must be that any team structure with only assignments with no cooperation is optimal, i.e., we are in case (ii) of the proposition. Thus, if $\nu^{coop}\in \supp(\lambda)$, then  $r_{\nu^{coop}}=0$. 

Since $d_G>\frac{\delta}{1-\delta}a_Gc_G$, there can be no optimal cooperation (Observation~\ref{obvs:full_specialization}). By optimality, there can also be no cooperation only on bad games; otherwise, the designer could improve their value by setting $r_\nu=1$ precluding cooperation. Thus, there is cooperation on both games given $\nu$. 

If $\nu(t_G)>\nu^{coop}(t_G)$, then there is no total cooperation as $\max_g \{d_g\} > \frac{\delta}{1-\delta} \left(\nu(t_G) a_G c_G + (1-\nu(t_G)) a_B c_B\right)$ by definition of $\nu^{coop}$ and Lemma~\ref{lem:equilibrium_characterization}. If $\nu(t_G)<\nu^{coop}(t_G)$, then there exists $x\in (0,1)$ such that $x+(1-x)\nu^{coop}(t_B)=\nu(t_B)$. Consider a team structure $\lambda'$ equal to $\lambda$ but where $\nu$ is replaced by $\nu^{B}$ with $\nu^{B}(t_B)=1$ and $r_B$ large enough to preclude cooperation (recall there cannot be $\nu^B\in \supp(\lambda)$ with $r_B$ low enough so that there is cooperation and so if $\nu^B\in \supp(\lambda)$, adopting the associated reshuffling rate precludes cooperation), and $\nu^{coop}$ with $r_{\nu^{coop}}=0$. In particular, let $\lambda'(\nu^{B})=\lambda(\nu^{B})+\lambda(\nu)x$ and $\lambda'(\nu^{coop})=\lambda(\nu^{coop})+\lambda(\nu)(1-x)$. Lemma~\ref{lem:equilibrium_characterization} implies there is total cooperation on $\nu^{coop}$; reshuffling ensures there is no cooperation on $\nu^{B}$. Note that there is strictly less cooperation on bad games while the amount of cooperation on good games is unchanged; i.e., $\nu(t_G)<\nu^{coop}(t_G)$ violates designer optimality, and so $\nu(t_G)=\nu^{coop}(t_G)$. 

Now, suppose that there is no cooperation on $\nu$. Since no team structure that only has positive weight on assignments that lead to no cooperation is optimal, it must be that the social value only considering teams assigned to $\nu^{coop}$ is strictly positive. This social value is linear mass $\lambda$ puts on $\nu^{coop}$, with the remaining mass contributing 0 to the social value. Hence, optimality requires that the mass of workers assigned to $\nu^{coop}$ is maximized; this implies that there cannot be a positive mass on both good and bad tasks assigned in $\nu$.
\end{proof}

\begin{proof}[Proof of Lemma~\ref{claim:pure}]
Consider any reactive assignment $m$ with reshuffling rate $r_m$,  and a strategy profile $\vec{\sigma}$. We construct a new strategy profile $\vec{\sigma}'$ that will also be an equilibrium that satisfies the conditions in the lemma and has weakly greater aggregate payoffs, and strictly so unless the initial equilibrium satisfies the conditions in the lemma.

Recall that $m$ can be equivalently viewed as a map from states $s\in \mathcal{S}\coloneqq (\mathcal{T}\times \{G,B,\emptyset\})^K$ into $\Delta(\mathcal{S})$. We adopt this view. Let $\mathcal{A}$ denote the set of states $s$ such that there exists a history $h\in \mathcal{H}$ such $(h,s)$ is reached with positive probability given $\vec{\sigma}$, and $\\vec{\sigma}(h,s)$ assigns a positive probability of cooperation for all players. Let $\vec{\sigma}'$ be defined by the following grim-trigger strategy: `Given any state $s \in (\mathcal{T}\times \{G,B,\emptyset\})^K$ and history $h\in\mathcal{H}$, cooperate if and only if $s \in \mathcal{A}$ unless any player deviated from this strategy in the past, in which case do not cooperate in either game.' 

As $m$ is unchanged, i.e., the distribution over states following each state, players' aggregate payoffs have increased using Properties~\ref{property:max} and~\ref{property:last}, which imply that the sum of on-path payoffs for each task and game pair increased, and strictly so unless on-path play already consisted of $C^n$ and $N^n$ only. 

What remains to show is that $\vec{\sigma}'$ is an equilibrium. To this end, first, note that whenever ${\sigma}'$ prescribes not to cooperate, it is optimal to do so since it is a strictly dominant action in the stage game (Property~\ref{property:dom}) and a player's future payoff from deviation profiles in the off-path continuation is weakly less than that from the on-path continuation which might contain some cooperation profiles (Property~\ref{property:max}). 

Next, consider any state $s$ in which players are prescribed to cooperate given $\vec{\sigma}'$. Since players are prescribed to cooperate, there exists a history $h$ such that given strategy profile $\vec{\sigma}$ and $(h,s)$, there is a positive probability of all players cooperating. We show that players not having a profitable deviation and not cooperating for sure given $\vec{\sigma}$ and $(h,s)$ implies they also do not have one given $\vec{\sigma}'$ and state $s$. We make three observations concerning the relative incentives to cooperate given $\vec{\sigma}$ and $(h,s)$ and  $\vec{\sigma}'$ and state $s$. First, using Property~\ref{property:star}, given $\vec{\sigma}$ and state $(h,s)$, a player not cooperating leads to an increase in their stage-game payoff of at least $d_{g}$ relative to cooperating; given $\vec{\sigma}'$ and $s$, that increase is exactly $d_g$. Thus the gains from deviating to $N$ in the stage game are higher for $\vec{\sigma}$ and $(h,s)$ than $\vec{\sigma}'$ and $s$ for all players.  We next show that there is some player $i$ who has less to lose in continuation payoffs in the former case than the later case.  Thus the equilibrium condition for $\vec{\sigma}$ for player $i$ implies the equilibrium condition for $\vec{\sigma}'$ for player $i$. First, as previously noted, players' future aggregate payoffs have increased given $\vec{\sigma}'$; thus, there exists a player $i$ whose future payoffs after cooperating are larger given $\vec{\sigma}'$ and $s$ than given $\vec{\sigma}$ and $(h,s)$. Next, a player's future payoffs after not cooperating are less given $\vec{\sigma}'$ and $s$ (in which case they face $N^{n-1}$ thereafter in the off-path play) than they are given $\vec{\sigma}$ and $(h,s)$ (this follows from Properties~\ref{property:max} and~\ref{property:minmax}). These two facts together mean that the loss in continuation payoff for player $i$ from deviating is larger in $\vec{\sigma}'$ than $\vec{\sigma}$.
Thus, since player $i$ cooperates with positive probability given  $\vec{\sigma}$ and $(h,s)$, player $i$ also cooperates given $\vec{\sigma}'$ and $s$. Since players have the same strategy, all players are incentivized to cooperate, and $\vec{\sigma}$ is indeed an equilibrium. 

To complete the proof, we need to show that the set of task and game pairs with cooperation is independent of the player-optimal equilibrium considered. Arguments analogous to those above show that an appropriately defined strategy profile with on-path cooperation on the set states for which there is cooperation in either player-optimal equilibrium is also an equilibrium and improves players' payoff, and strictly so unless the sets of task and game pairs with cooperation were the same across equilibria. We omit the details.
\end{proof}

\begin{proof}[Proof of Proposition~\ref{prop:when}]
    Consider a class of reactive assignments $m$ parameterized by $w_{gg'}\in[0,1]$ for $g,g'\in\{G,B\}$ defined as follows. Consider any state $[(t_\tau, g_\tau)]_{\tau=0}^K$ with current task $t_K=t_{g}$; then let $m([(t_\tau, g_\tau)]_{\tau=0}^K)(t_g')=w_{gg'}$ for $w_{gg'}$ such that $w_{gG}+w_{gB}=1$ for all $g \in \{G,B\}$. Note that for $w_{Gg}=w_{Bg}=\nu^{coop}(t_g)$, then  $m$ corresponds to the same static assignment as $\nu^{coop}$.

    For interior transition probabilities, let $\overline{m}(t_G)$ denote the unique steady state mass on the good task; we have $\overline{m}(t_G)=\frac{w_{BG}}{w_{BG}+w_{GB}}$. For $w_{Gg}=w_{Bg}=\nu^{coop}(t_g)$, let this mass be denoted by $\overline{m}^{coop}(t_G)$. Note that all $w_{BG},w_{GB}$ (with $w_{BB}=1-w_{BG},w_{GG}=1-w_{GB}$) such that 
    \begin{equation}\label{eqn:normalization}
        w_{GB}=w_{BG}\frac{1-\overline{m}^{coop}(t_G)}{\overline{m}^{coop}(t_G)}
    \end{equation} 
    lead to the same steady state mass on the good task, $\overline{m}^{coop}(t_G)$, and therefore also on the bad task $\overline{m}^{coop}(t_B)$. As a result, all such $w_{BG},w_{GB}$ imply the same average benefit from future cooperation in the two assignments $m$ and $\nu^{coop}$. 
    
    Let $\overline{g},\underline{g}$ be such that $d_{\overline{g}}>d_{\underline{g}}$, i.e., it is harder to incentivize cooperation in game $\overline{g}$. The conditions of the proposition imply $a_Gc_G<a_Bc_B$. Therefore, in an assignment $m$ that is not the same static assignment as $\nu^{coop}$ the average benefit from future cooperation can be made strictly larger when teams are in game $\overline{g}$ than when they are in game $\underline{g}$ by sending them more frequently to the bad task when in $\overline{g}$, i.e., by increasing $w_{\overline{g}B}$.  Furthermore, this can be done while maintaining the same unconditional benefit from future cooperation in this modified $m$ and original $\nu^{coop}$ by decreasing $w_{\underline{g}B}$ according to Equation~\ref{eqn:normalization}. By continuity and if $d_{\overline{g}}$ is strictly less than $d_{\underline{g}}$, there are thus values of $w_{BG},w_{GB}$ (with $w_{BB}=1-w_{BG},w_{GG}=1-w_{GB}$) such that the benefit from future cooperation is strictly greater than the deviation payoff in both tasks, which allows the designer to strictly improve the social value, as in the proof of Lemma~\ref{lem:design_opt}.
\end{proof}

\begin{proof}[Proof of Lemma~\ref{lem:design_opt}]
Consider any reactive team structure $\lambda$, $m \in \supp(\lambda)$ with reshuffling rate $r_m$ and steady-state distribution $\overline{m}$, and player-optimal equilibrium $\vec{\sigma}_m$. 

Let $\mathcal{S}^m\subseteq \mathcal{S}\coloneqq (\mathcal{T}\times \{G,B,\emptyset\})^K$ be the set of states with positive mass given steady state $\overline{m}$. Let $\mathcal{A}^C\subseteq \mathcal{S}^M$ be the set of states in which players cooperate on-path (which is well-defined given Lemma~\ref{claim:pure}). And let $\mathcal{A}^{\emptyset}$ be the set of states $s^{\emptyset}$ such that there exists a state $s  \in \mathcal{A}^C$ that only differs in the $K$-th game entry from state $s$; in particular, state $s$ has a game $g\in\{G,B\}$ and state $s^\emptyset$ has no game. Let $\mathcal{A}\coloneqq\mathcal{A}^C\cup \mathcal{A}^\emptyset$; if $\mathcal{A}=\emptyset$ or $\mathcal{A}^C\cup \mathcal{A}^\emptyset=\mathcal{S}^m$, then the claim in the lemma holds. 

Thus, suppose $\emptyset \subset \mathcal{A} \subset\mathcal{S}^m$. We construct an alternative reactive team structure $\lambda^*$ that improves the designer's value. To this end, we define two reactive assignments, $m^C$ and $m^{NC}$. 

For any $s,s' \in \mathcal{S}$, if  $s \in \mathcal{A}$, and  
\begin{itemize}
    \item if $s' \in \mathcal{A}$, then 
    let $\mathcal{S}^l_{ss'}$ be the set of length $(l+1)$ state sequences  starting at $s$ (a state with cooperation), passing through states in $\mathcal{S}\setminus\mathcal{A}$ (states without cooperation), and ending in $s'$ (another state with cooperation), i.e., sequences $(s_0,\ldots,s_{l+1})$ such that $s_0=s,s_{l+1}=s'$ and $(s_1,\ldots,s_{l})\in(\mathcal{S}\setminus\mathcal{A})^{l}$. Define 
    $$m^C(s'|s)=m(s'|s)+\sum_{l=1}^\infty \left[\sum_{(s_0,\dots,s_{l+1})\in \mathcal{S}^l_{ss'}} \left[\prod_{i=1}^{l+1} m(s_i|s_{i-1})\right]\right],$$
    \item if $s' \not\in \mathcal{A}$, then $m^C(s'|s)=0$, 
\end{itemize}
and if $s \not\in \mathcal{A}$, and 
\begin{itemize}
    \item if $s' \in \mathcal{A}$, then $m^C(s'|s)=0$,
    \item if $s' \not\in \mathcal{A}$, then 
    let $\mathcal{S}^l_{ss'}$ be the set of length $(l+1)$ state sequences  starting at $s$ (a state without cooperation), passing through states in $\mathcal{A}$ (states with cooperation), and ending in $s'$ (another state without cooperation), i.e., sequences $(s_0,\ldots,s_{l+1})$ such that $s_0=s,s_{l+1}=s'$ and $(s_1,\ldots,s_{l})\in(\mathcal{A})^{l}$. Define 
    $$m^C(s'|s)=m(s'|s)+\sum_{l=1}^\infty \left[\sum_{(s_0,\dots,s_{l+1}) \in \mathcal{S}_{ss'}^l} \left[\prod_{i=1}^{l+1} m(s_i|s_{i-1})\right]\right],$$
\end{itemize}
and let $m^{NC}=m^C$,  with $\overline{m}^C$  chosen to have mass only on states in $\mathcal{A}$ and $\overline{m}^{NC}$ to have mass only on states in $\mathcal{S}\setminus\mathcal{A}$.

In words, $m^{NC}$ and $m^C$ are ``skipping'' states: either the Markov chain is in $\mathcal{A}$, in which case it skips states in $\mathcal{S}\setminus \mathcal{A}$ with no cooperation, or it is in $\mathcal{S}\setminus \mathcal{A}$, in which case it skips states in $\mathcal{A}$ with cooperation. 

As a result, for appropriately chosen $\lambda'(m^{C}),\lambda'(m^{NC})$, where $\lambda'(m^{C})+\lambda'(m^{NC})=\lambda(m)$, and steady states for $m^C$ and $m^{NC}$, we have $\overline{m}=\overline{m}^C+\overline{m}^{NC}$. 
Let $r_{m^C}=0$, and note that assuming there is cooperation in every state $s \in \mathcal{A}^C$ given reactive assignment $m^C$, the future benefit from cooperation strictly increased, and so players are indeed strictly incentivized to cooperate. Furthermore, let $r_{m^{NC}}=1$, note that there is no cooperation on any state $s\in \mathcal{A}^{NC}$ given reactive assignment $m^{NC}$. Thus, the reactive team structure $\lambda'$ covers each task at the same rate as $\lambda$ and has the same value to the designer. 

For a strict improvement, consider a state $s_G \in \mathcal{A}^C$, where the $K$-th game entry is a good game, and state $s_G^\emptyset$, which equals $s_G$ with the exception of the $K$-th game entry which is no game. Such states exist; otherwise, there is only on-path cooperation on bad games in the reactive assignment $m$, which violates designer optimality. Define two new reactive assignments, $m^*$ and $m_B$. For $s=s_G,s_G^\emptyset$, let 
$$
m^*(\cdot|s)=(1-\epsilon)m^C(\cdot|s)+\epsilon (a_G \delta_{s_G}+(1-a_G)\delta_{s_G^\emptyset}),
$$
where $\delta_{s_G},\delta_{s_G^\emptyset}$ are the Dirac measures of $s_G$ and $s_G^\emptyset$, respectively, and $\epsilon\in (0,1]$, and $m^*(\cdot|s)=m^C(\cdot|s)$ for all $s\neq s_G,s_G^\emptyset$.

In words, reactive assignment $m^*$ equals $m^C$ with the exception that there is an additional chance of remaining on the good task for some history. Note that there exists a steady state distribution $\overline{m}^*$ with full support on $\mathcal{A}$ and strictly more weight on good tasks relative to $\overline{m}^C$. With $\epsilon>0$ small enough and $r_{m^*}=0$ there is still on-path cooperation everywhere (since players were strictly incentivized to cooperate everywhere given $m^C$).

Reactive assignment $m_B$ is a constant assignment to the bad task with $r_{m_B}=1$. (Formally, it is defined using the equivalent definition of a reactive assignment as a function from $(\mathcal{T}\times \{G,B,\emptyset\})^K$ to $\Delta(\mathcal{T})$ as $m_B=\delta_{t_B}$, where $\delta_{t_B}$ is the Dirac measure of the bad task.) 

Since the steady state of $m^*$ has strictly more weight on good tasks than $m^C$ and the steady state of $m_B$, which has full weight on bad tasks, has strictly less weight on good tasks than $m^C$ ($m^C$ cannot have full weight on good tasks; otherwise, there would be no cooperation as we $d_G>\frac{\delta}{1-\delta}a_Gc_G$ in this section), there exist some weights $\lambda^*(m^*),\lambda^*(m_B)$, such that the combination of $m^*$ and $m_B$ covers tasks at the same rate as $m^C$. It follows that reactive team structure $\lambda^*$, defined by the weights above for reactive assignments $m^*$ and $m_B$, and otherwise equal to $\lambda'$, covers tasks at the same rate as $\lambda'$ (and hence $\lambda$); however, since there is relatively more weight on good tasks in $m^*$ than $m^C$, and players always cooperate, and players do not cooperate given $m_B$, the designer's value strictly increased. Hence, $\emptyset \subset \mathcal{A} \subset\mathcal{S}^m$ leads to a contradiction, completing the proof.
\end{proof}

\begin{proof}[Proof of Proposition~\ref{prop:main}]
Consider $m \in \supp(\lambda)$, with reshuffling rate $r_m$ and steady state $\overline{m}$, with fully cooperative teams (Lemma~\ref{lem:design_opt}). 

We first show by contradiction that players must be indifferent between cooperating and not cooperating when playing a good game as claimed in the first implication of the proposition.  We define an alternative reactive assignment $m'$ with the same designer value, where continuation after a good game is independent of the history. 
Let $\mathcal{S}\coloneqq \left(\mathcal{T}\times \{G,B,\emptyset\}\right)^K$ and let $\mathcal{A}_G \subseteq \mathcal{S}$ be the set of states $s_G$ with a good game as its $K$-th game entry. For all $s_G \in \mathcal{A}_G$, let
$$
m'(\cdot|s_G)=\frac{\sum_{s' \in \mathcal{A}_G}\overline{m}({s'})m(\cdot|s')}{\sum_{s'' \in \mathcal{A}_G}\overline{m}({s''})},
$$
and for all $s \not\in \mathcal{A}_G$, let $m'(\cdot|s)=m(\cdot|s)$. 
Since $\sum_{s\in \mathcal{S}} \overline{m}(s)m(\cdot|s)=\sum_{s\in \mathcal{S}}  \overline{m}(s)m'(\cdot|s)$, $\overline{m}$ continues to be a steady state. Assuming players continue to cooperate on-path in all states, it is indeed optimal to do so for any state with a bad game as the $K$-th game entry since $d_B<\frac{\delta}{1-\delta}a_Gc_G$. Furthermore, note that as the continuation after a good game is now the same across states, the future benefit from cooperation following a good game is now the same across states---it is the weighted average; hence, players indeed continue to cooperate on good games. Together with the observation that the steady state did not change, this implies there is no change to the designer's value between $m$ and $m'$.  However, if players were strictly incentivized to cooperate in $m$, they would be strictly incentivized to cooperate given $m'$ and $r_{m'}=0$, violating designer optimality as in the proof of Lemma~\ref{lem:design_opt}.

We next define another reactive assignment $m^*$ that satisfies the structure discussed in the second two implications of the proposition and has strictly better value than $m'$ unless $m'$ also satisfies that structure, thus proving such structure is necessary as claimed. Given $m'$, consider a team that faced a good game in some period $l$. For $k\geq1$, define two events: the first is that a bad game is faced in period $l+k$ and the second is that no good game was faced in the periods between $l$ and $l+k$. Let $p_k$ denote the probability of the intersection of those two events. Let $N_B\coloneqq \left\lfloor \frac{\sum_{k=1}^\infty p_k}{a_B}\right\rfloor$, i.e., $N_B$ denotes the integer part of the number of bad tasks that would result in the same expected number of bad games. 
Consider states $s$ with a good game as the game entry in the $l$-th position for some $l\geq K-N_B$. Then, let 
$$
m^*(\cdot|s)=\delta_{t_B},
$$
where $\delta_{t_B}$ is the Dirac measure of task $t_B$ ($\delta_{t_G}$ is defined analogously). Consider states $s$ with a good game as the game entry in the $K-N_B-1$-th position. Let $x\coloneqq \frac{\sum_{k=1}^\infty p_k}{a_B}-N_B$. 
Then, let 
$$
m^*(\cdot|s)=x\delta_{t_B}+(1-x)\delta_{t_G}.
$$
Consider states $s$ such that the $K$-th task entry is $t_G$ and the $K$-th game entry is no game.  Then, let 
$$
m^*(\cdot|s)=\delta_{t_G}.
$$
Note, by design, $m^*$ satisfies the structure claimed in the lemma. 

To argue about the designer's value given $m^*$, first note that by design, there exists a steady state of $m^*$ such that the distribution over tasks is the same as in the steady state of $m'$.  We argue that $m^*$ has full cooperation, thus implying the designer's value is the same as in $m'$.

Assuming players continue to cooperate on-path in all states, as before, it is indeed optimal to do so for any state with a bad game as the $K$-th game entry since $d_B<\frac{\delta}{1-\delta}a_Gc_G$. 

Next, consider a state $s$ with a good game as the $K$-th game. Suppose given $m^*$ (resp.\ $m'$), the benefit from cooperation following any other state with a good game as the $K$-th game is $V_G^*$ (resp.\ $V_G'$). 
Note that $a_Bc_B + \delta a_G(c_G+V_G^*) > a_G(c_G+V_G^*)$ (resp.\ $a_Bc_B + \delta a_G(c_G+\delta V_G^') > a_G(c_G+\delta V_G^')$), since the conditions of the proposition imply $a_Gc_G<a_Bc_B$, and since $V_G^* \leq \frac{1}{1-\delta}a_Bc_B$ (resp.\ $V_G^' \leq \frac{1}{1-\delta}a_Bc_B$). Thus, the future benefit from cooperation following $s$ strictly increases if players are assigned to bad tasks before good tasks.
Given $m^*$, bad tasks arrive as soon as possible; specifically, $\sum_{k=1}^Lp_k \leq \sum_{k=1}^Lp^*_k$ for all $L\geq 1$, where $p^*_k$ is defined analogously to $p_k$ given $m^*$. 
Hence, assuming the benefit from cooperation following any state other than $s$ with a good game as the $K$-th game is $V_G^'$ in $m^*$ (rather than $V_G^*$), the benefit from cooperation after $s$ is larger given $m^*$ than $m'$. Furthermore, the comparison is strict unless $p_k=p^*_k$ (implying $m'$ and $m$ have the structure given in the proposition). The increased incentive to cooperate in $s$, in turn, implies $V_G^* \geq V_G'$, and, again, strictly so unless $p_k=p^*_k$; but strict incentives preclude designer optimality as in the proof of Lemma~\ref{lem:design_opt}.   

For the last statement in the proposition, suppose all designer-optimal reactive team structures $\lambda$ involve some reactive assignment with only fully cooperative teams. Consider $m,m'\in\supp(\lambda)$ such that $m$ is fully cooperative and $m'$ is fully non-cooperative. Note that there is a choice of $m$ that has strictly positive social value for the designer (as otherwise the designer gets zero social value from $\lambda$ and hence a static structure with full remixing---and hence no cooperation---is also optimal violating the assumption that all optimal structures involve some cooperation).  Thus, if there is a positive mass on both tasks in the steady state associated with $m'$, then the reactive assignment with cooperation $m$ can receive more mass, which strictly increases the designer's value. Furthermore, suppose there are two fully non-cooperative assignments $m_G$ and $m_B$ with mass only on the good (resp.\ bad) task and $\lambda(m_G),\lambda(m_B)>0$. Then the designer can improve by moving teams from $m_B$ and $m_G$, which yield zero social value, to $m$, which yields positive social value, in a proportion that maintains the task assignment constraint (\ref{eq:reactiveconstraint}).
\end{proof}

\begin{proof}[Proof of Proposition~\ref{prop:main_g_unob}] 
Consider $m \in \supp(\lambda)$, with reshuffling rate $r_m$ and steady state $\overline{m}$, with fully cooperative teams (Lemma~\ref{lem:design_opt}). As a preliminary step, we consider an alternative reactive assignment $m'$ with the same designer value, where continuation after a good task is independent of the history. 

Let $\mathcal{S}\coloneqq \left(\mathcal{T}\times \{G,B,\emptyset\}\right)^K$ and let $\mathcal{A}_G \subseteq \mathcal{S}$ be the set of states $s_G$ with a good task as its $K$-th game entry. For all $s_G \in \mathcal{A}_G$, let
$$
m'(\cdot|s_G)=\frac{\sum_{s' \in \mathcal{A}_G}\overline{m}({s'})m(\cdot|s')}{\sum_{s'' \in \mathcal{A}_G}\overline{m}({s''})},
$$
and for all $s \not\in \mathcal{A}_G$, let $m'(\cdot|s)=m(\cdot|s)$. 
Since $\sum_{s\in \mathcal{S}} \overline{m}(s)m(\cdot|s)=\sum_{s\in \mathcal{S}}  \overline{m}(s)m'(\cdot|s)$, $\overline{m}$ continues to be a steady state. Assuming players continue to cooperate on-path in all states, it is indeed optimal to do so for any state with a bad game as the $K$-th game entry since $d_B<\frac{\delta}{1-\delta}a_Gc_G$. Furthermore, note that as the continuation after a good game is now the same across states, the future benefit from cooperation following a good game is now the same across states---it is the weighted average; hence, players indeed continue to cooperate on good games. Also, note that players must have been indifferent between cooperating or not in good games given $m$; otherwise, they would be strictly incentivized to cooperate given $m'$ and $r_{m'}=0$, violating designer optimality as in the proof of Lemma~\ref{lem:design_opt}. Lastly, note that there is no change to the designer's value. 

We next define another reactive assignment $m^*$. Given $m'$, consider a team assigned to the good task in some period $l$. For $k\geq1$, define two events: the first is that a bad game is faced in period $l+k$ and the second is that the team has not been assigned the good task in the periods between $l$ and $l+k$. Let $p_k$ denote the probability of the intersection of those two events. Let $N_B\coloneqq \left\lfloor \frac{\sum_{k=1}^\infty p_k}{a_B}\right\rfloor$, i.e., $N_B$ denotes the integer part of the number of bad tasks that would result in the same expected number of bad games. 
Consider states $s$ with a good task as the task entry in the $l$-th position for some $l\geq K-N_B$. Then, let 
$$
m^*(\cdot|s)=\delta_{t_B},
$$
where $\delta_{t_B}$ is the Dirac measure of task $t_B$ ($\delta_{t_G}$ is defined analogously). Consider states $s$ with a good task as the task entry in the $K-N_B-1$-th position. Let $x\coloneqq \frac{\sum_{k=1}^\infty p_k}{a_B}-N_B$. 
Then, let 
$$
m^*(\cdot|s)=x\delta_{t_B}+(1-x)\delta_{t_G}.
$$
Note, by design, there exists a steady state of $m^*$ such that the distribution over tasks is the same as in the steady state of $m'$.

Assuming players continue to cooperate on-path in all states, as before, it is indeed optimal to do so for any state with a bad game as the $K$-th game entry since $d_B<\frac{\delta}{1-\delta}a_Gc_G$. Next, fixing the future benefit from cooperation following a good game, note that the future benefit from cooperation strictly increases if a good game is delayed by a period and players are assigned to a bad task for one period instead; this follows as the conditions of the proposition imply $a_Gc_G<a_Bc_B$. Thus, players indeed continue to cooperate on good games. Furthermore, unless $m'$ already had the structure given in the proposition, players are now strictly incentivized to cooperate, precluding designer optimality as in the proof of Lemma~\ref{lem:design_opt}. 

For the last statement in the proposition, suppose all designer-optimal reactive team structures $\lambda$ involve some reactive assignment with only fully cooperative teams. Consider $m,m'\in\supp(\lambda)$ such that $m$ is fully cooperative and $m'$ is fully non-cooperative. Note that there is a choice of $m$ that has strictly positive social value for the designer (as otherwise the designer gets zero social value from $\lambda$ and hence a static structure with full remixing -- and hence no cooperation -- is also optimal violating the assumption that all optimal structures involve some cooperation).  Thus, if there is a positive mass on both tasks in the steady state associated with $m'$, then the reactive assignment with cooperation $m$ can receive more mass, which strictly increases the designer's value. Furthermore, suppose there are two fully non-cooperative assignments $m_G$ and $m_B$ with mass only on the good (resp.\ bad) task and $\lambda(m_G),\lambda(m_B)>0$. Then the designer can improve by moving teams from $m_B$ and $m_G$, which yield zero social value, to $m$, which yields positive social value, in a proportion that maintains the task assignment constraint (\ref{eq:reactiveconstraint}).

\end{proof}

\newpage

\begin{center}
{\Large {Online Appendix for}}\\[1em]
{\Large {``Interactions across multiple games: 
cooperation, corruption, and organizational design''}}\\[1.2em]

{\large
{Jonathan Bendor, Lukas Bolte, Nicole Immorlica, Matthew O.\ Jackson}
}
\end{center}

\section*{Comparative Statics in Payoffs and Design Implications}\label{sec:changing_payoffs}

Given that payoffs in one game have spillover effects on the other game, organizations can use bonuses and penalties to alter game payoffs and thereby induce optimal (or better) cooperation. Here, we study how such changes in payoffs affect the potential for cooperation. Doing so also helps us understand why some organizations or countries have optimal cooperation while others are stuck with tolerating bad cooperation or having no cooperation at all. 

We identify a fundamental difference in the incentives provided by changing the benefits from cooperation (the $c_g$'s, e.g., giving bonuses for team production) versus the temptations to deviate (the $d_g$'s, e.g., rewarding whistle-blowing). When examining a single repeated game in isolation, incentives depend only on the \emph{relative} size of deviation temptations to cooperation benefits---the ratio $d_g/c_g$. There exists a subgame-perfect equilibrium in which players cooperate perpetually if and only if $ \frac{d_g}{c_g}\leq \frac{\delta}{1-\delta}p_g$. This is not so with more than one game: cooperating in one game provides incentives to cooperate in another, but deviation payoffs affect the incentives in one game at a time. 

This asymmetry makes the comparative statics of multiple games differ from those of standard repeated games. For example, one might expect that improving the benefits from good cooperation would enable cooperation in the good game only, possibly with (some) reshuffling. But this is not always so. To see why, consider the case of $d_G\geq d_B$. Then, whenever team members can sustain cooperation in good situations, the temptation to deviate from bad cooperation is also less than the continuation value; hence, they can collude in bad situations as well. This is true regardless of the level of payoffs to cooperation. Thus, in order for policies that increase $c_G$ to impact optimal cooperation at all, deviation temptations must already be lower in the good game than the bad. In that case, i.e., when  $d_G<d_B$, optimal cooperation can be enabled via sufficient increases to $c_G$. 

Panels~(a) and~(b) of Figure~\ref{fig:comp_stats} depict how changing the benefits from cooperation in the good game can make optimal cooperation feasible. Reshuffling decreases the effective discount factor, i.e., $(1-r)\delta$, moving the outcome of the game from $x_{\text{initial}}=(\frac{\delta}{1-\delta}p_B,\frac{\delta}{1-\delta}p_G)$ in direction of the origin. Panel~(a) displays an initial situation where $d_B>d_G$ but $c_G$ is too low to reshuffle teams appropriately and achieve optimal cooperation.  Increasing $c_G$ to some $c_G'$ rotates the line delineating the region of total cooperation counter-clockwise. For a large enough increase (see Panel~(b)), reshuffling teams and inducing optimal cooperation is feasible---there exists a reshuffling rate $r$ such that  $x_{\text{reshuffled}}=(\frac{(1-r)\delta}{1-(1-r)\delta}p_B,\frac{(1-r)\delta}{(1-r)1-\delta}p_G)$ is in the region with optimal cooperation.

\begin{figure}[ht]
\small
\captionsetup{width=0.9\textwidth}
    \centering
    \begin{subfigure}[t]{0.5\textwidth}
        \centering
        \begin{tikzpicture}
        \def \xscale {.25} 
        \def \yscale {.2}
            \draw[thick,->] (0*\xscale,0*\yscale) -- (20.5*\xscale,0*\yscale) node[anchor=north west] {$\frac{\delta^e}{1-\delta^e}p_B$};
            \draw[thick,->] (0*\xscale,0*\yscale) -- (0*\xscale,20.5*\yscale) node[anchor=south east] {$\frac{\delta^e}{1-\delta^e}p_G$};
            \draw[thick] (0*\xscale,17.5*\yscale) node [anchor=east] {$d_B/c_G$} --(17.5*\xscale,0*\yscale) node [anchor=north] {$d_B/c_B$};
            \draw[thick] (7.5*\xscale,10*\yscale) -- (0*\xscale,10*\yscale) node [anchor=east] {$d_G/c_G$};
            \draw (6*\xscale,16*\yscale) node[anchor=south,align=center] {cooperation only\\in good game};
            \draw[->] (6*\xscale,16*\yscale) -- (3*\xscale,12*\yscale);
            \filldraw[thick] (12*\xscale,12*\yscale) circle (.05 cm) node[anchor=north west] {$x_{\text{initial}}$};
            \draw (5*\xscale,2.5*\yscale) node[anchor=south,align=center] {no \\ cooperation};
            \draw (14*\xscale,16*\yscale) node[align=center] {total \\ cooperation};
        \end{tikzpicture}
        \caption{Initially: $d_B>d_G$ but $c_G$ is too low for optimal cooperation.}
        \label{subfig:cg_initial}
    \end{subfigure}%
    ~ 
    \begin{subfigure}[t]{0.5\textwidth}
        \centering
        \begin{tikzpicture}
        \def \xscale {.25} 
        \def \yscale {.2}
            \draw[thick,->] (0*\xscale,0*\yscale) -- (20.5*\xscale,0*\yscale) node[anchor=north west] {$\frac{\delta^e}{1-\delta^e}p_B$};
            \draw[thick,->] (0*\xscale,0*\yscale) -- (0*\xscale,20.5*\yscale) node[anchor=south east] {$\frac{\delta^e}{1-\delta^e}p_G$};
            \draw[thick,dashed] (0*\xscale,17.5*\yscale) node [anchor=east] {$d_B/c_G$} --(17.5*\xscale,0*\yscale) node [anchor=north] {$d_B/c_B$};
            \draw[thick,dashed] (7.5*\xscale,10*\yscale) -- (0*\xscale,10*\yscale) node [anchor=south east] {$d_G/c_G$};
            \draw[thick] (0*\xscale,17.5/1.85*\yscale) node [anchor=north east] {$d_B/c_G'$} --(17.5*\xscale,0*\yscale);
            \draw[thick] (7.5*\xscale,10/1.85*\yscale) -- (0*\xscale,10/1.85*\yscale) node [anchor=east] {$d_G/c_G'$};
            \filldraw[thick,dashed] (12*\xscale,12*\yscale) circle (.05 cm) node[anchor=north west] {$x_{\text{initial}}$} -- (6*\xscale,6*\yscale) circle (.05 cm) node[anchor=north west] {$x_{\text{reshuffled}}$};
        \end{tikzpicture}
        \caption{Increase: $c_G$ to $c_G'$}
        \label{subfig:cg_increased}
    \end{subfigure}
    \begin{subfigure}[t]{0.5\textwidth}
        \centering
        \begin{tikzpicture}
        \def \xscale {.25}
        \def \yscale {.2}
            \draw[thick,->] (0*\xscale,0*\yscale) -- (20.5*\xscale,0*\yscale) node[anchor=north west] {$\frac{\delta^e}{1-\delta^e}p_B$};
            \draw[thick,->] (0*\xscale,0*\yscale) -- (0*\xscale,20.5*\yscale) node[anchor=south east] {$\frac{\delta^e}{1-\delta^e}p_G$};
            \draw[thick] (0*\xscale,17.5*\yscale) node [anchor=east] {$d_G/c_G$} --(17.5*\xscale,0*\yscale) node [anchor=north] {$d_G/c_B$};
            \draw[thick] (12.5*\xscale,5*\yscale) -- (12.5*\xscale,0*\yscale) node [anchor=north] {$d_B/c_B$};
            \draw (14*\xscale,13*\yscale) node[align=center] {total \\ cooperation};
            \draw (5*\xscale,4*\yscale) node[anchor=south,align=center] {no \\ cooperation};
            \draw (18*\xscale,4*\yscale) node[anchor=south,align=center] {cooperation only\\in bad game};
            \draw[->] (18*\xscale,4*\yscale) -- (13.75*\xscale,1.5*\yscale);
            \filldraw[thick, dashed] (8*\xscale,15*\yscale) circle (.05cm) node[anchor=south west] {$x_{\text{initial}}$};
        \end{tikzpicture}
        \caption{Initially: $d_G>d_B$.}
        \label{subfig:dg_initial}
    \end{subfigure}%
    ~ 
    \begin{subfigure}[t]{0.5\textwidth}
        \centering
        \begin{tikzpicture}
        \def \xscale {.25}
        \def \yscale {.2}
            \draw[thick,->] (0*\xscale,0*\yscale) -- (20.5*\xscale,0*\yscale) node[anchor=north west] {$\frac{\delta^e}{1-\delta^e}p_B$};
            \draw[thick,->] (0*\xscale,0*\yscale) -- (0*\xscale,20.5*\yscale) node[anchor=south east] {$\frac{\delta^e}{1-\delta^e}p_G$};
            \draw[thick, dashed] (0*\xscale,17.5*\yscale) node [anchor=east] {$d_G/c_G$} --(17.5*\xscale,0*\yscale) node [anchor=north] {$d_G/c_B$};
            \draw[thick, dashed] (12.5*\xscale,5*\yscale) -- (12.5*\xscale,0*\yscale);
            \draw[thick] (0*\xscale,12.5*\yscale) node [anchor=east] {$d_B/c_G$} --(12.5*\xscale,0*\yscale) node [anchor=north] {$d_B/c_B$};
            \filldraw[thick, dashed] (8*\xscale,15*\yscale) circle (.05cm) node[anchor=south west] {$x_{\text{initial}}$} -- (8*.5*\xscale,15*.5*\yscale) circle (.05cm) node[anchor=north] {$x_{\text{reshuffled}}$};
            \draw[thick] (5.5*\xscale,7*\yscale) -- (0*\xscale,7*\yscale) node [anchor=east] {$d_G'/c_G$};
        \end{tikzpicture}
        \caption{Decrease: $d_G$ to $d_G'$.}
        \label{subfig:dg_decreased}
    \end{subfigure}%
    \caption{Changing benefits from cooperation, $c_G$, and deviation temptations, $d_G$.}
    \label{fig:comp_stats}
\end{figure}
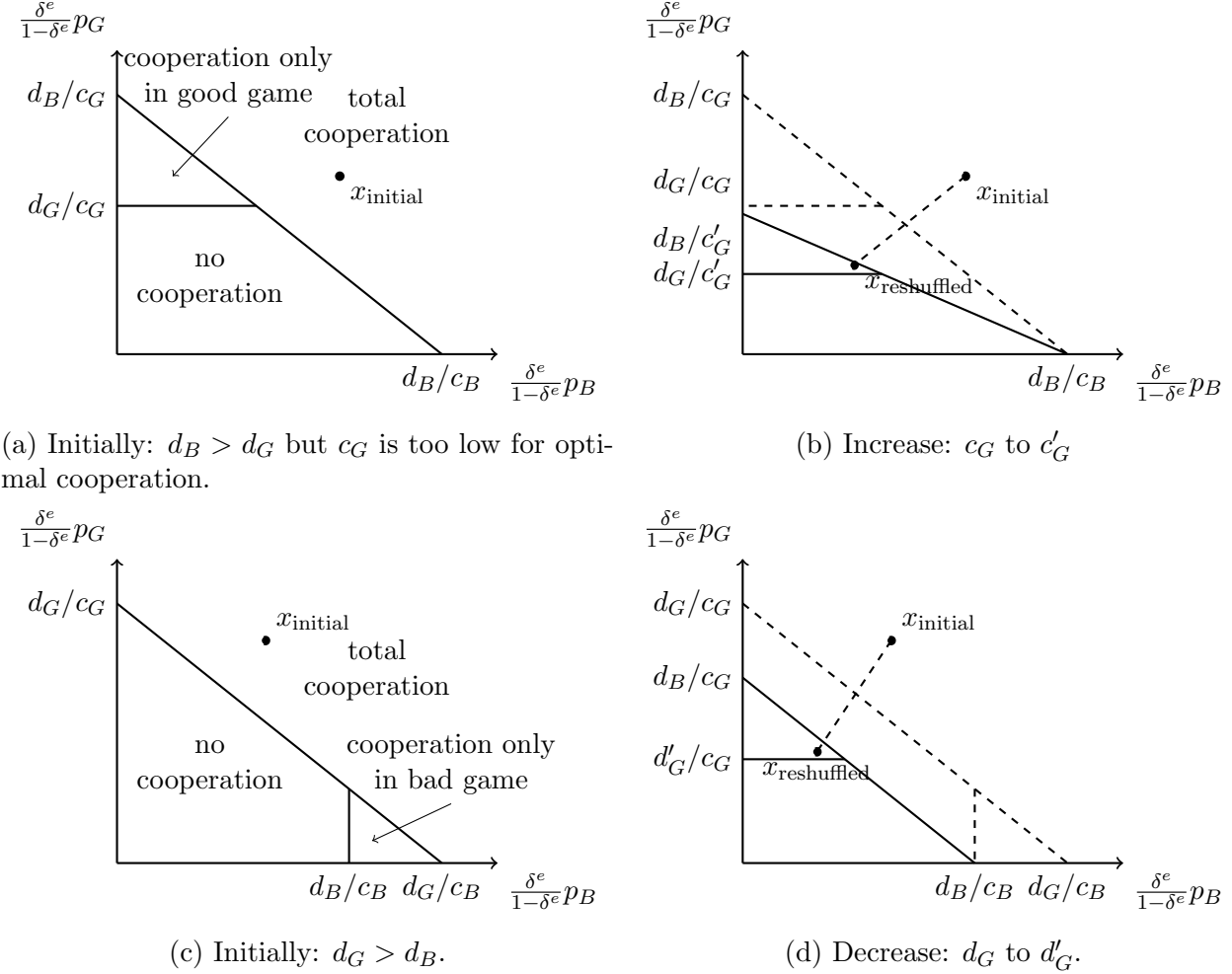

If $d_G>d_B$, optimal cooperation can be achieved only by changing the temptations to deviate. These include introducing rewards to whistle-blowing (increasing $d_B$) or reducing the cost of cooperating in good games (decreasing $d_G$). If $d_G$ can be pushed sufficiently close to zero, then optimal cooperation is feasible with optimal reshuffling. Panels~(c) and~(d) of Figure~\ref{fig:comp_stats} depict such a case. Panel~(c) displays an initial situation with $d_G>d_B$, where cooperation only in the good game is infeasible. We then consider decreasing $d_G$ to some $d_G'$ (see Panel (d)). For $d_G'$ smaller than $d_B$, (i) the region with total cooperation shifts inwards, and (ii) a region where optimal cooperation is the equilibrium outcome appears. If $d_G'$ is small enough so that reshuffling can maneuver the outcome of the game into that region (as displayed), optimal cooperation again becomes feasible ($d_G'<d_B\left(\frac{p_Gc_G}{p_Gc_G+p_Bc_B}\right)$, by Proposition~\ref{prop:optimal_reshuffling}). 

The temptations to deviate in the good game may be inherent to the situation; hence, reducing $d_G$ may be difficult. Or reshuffling teams may be infeasible, e.g., because the organization is small or only a few players, such as an organization's top executives, are relevant. In such circumstances, whistle-blowing rewards can hinder collusion in the bad game. In fact, reshuffling is unnecessary for optimal cooperation if the temptation to deviate in bad games can be increased sufficiently: simply set $d_B>c_G\frac{\delta}{1-\delta}p_G+c_B\frac{\delta}{1-\delta}p_B$. Then, as long as condition~\ref{cond:absolute} in Proposition~\ref{prop:optimal_reshuffling} is satisfied (cooperation is sustainable in the good game), optimal cooperation is feasible without reshuffling. Notice also that in contrast to changing the benefit of good cooperation, changing the deviation temptations does not raise the organization's costs because these rewards and risks are not incurred on the equilibrium path.\footnote{For example, suppose that the organization pays out $c_G$ when there is cooperation (we still assume that actions taken in the games are not contractible). Then, increasing $c_G$ to $c_G'$ to induce optimal cooperation is organizationally optimal if and only if $p_BV_B+p_G(c_G'-c_G)>0$; that is, when the loss of undesired cooperation, weighted by its social value, exceeds the cost of increasing the benefit from good cooperation. In contrast, increasing $d_B$ works only via deterrence so that rewards for whistle-blowing are not paid in equilibrium.}

Do these results imply that an organization always welcomes large rewards for whistle-blowing, i.e., large $d_B$? No. While increasing $d_B$ relaxes two of the three conditions, the third might still not be satisfied: the requirement that cooperation is sustainable in the good games is independent of $d_B$. If this condition fails and whistle-blower rewards are introduced, total cooperation becomes more difficult, possibly yielding no cooperation whatsoever. Therefore, whistle-blower rewards are useful only when good games arrive with high enough probability to sustain cooperation in the good game alone, or deviation payoffs or benefits from cooperation in the good game can be adjusted; that is, a combination of multiple changes in payoffs is needed.

\end{document}